\documentclass[11pt,reqno]{amsart}
\usepackage{amsfonts,amsmath,amssymb}
\newtheorem{thm}{Theorem}[section]
\newtheorem{proposition}[thm]{Proposition}
\newtheorem{corollary}[thm]{Corollary}
\newtheorem{lemma}[thm]{Lemma}

\newtheorem{remark}[thm]{Remark}

\usepackage{color}
\usepackage{graphicx}
\usepackage{epstopdf}

\title[Jack--Laurent symmetric functions]{Jack--Laurent symmetric functions}
\author{ A.N. Sergeev}\address{Department of Mathematics, Saratov State University, Astrakhanskaya 83, Saratov, 410012, Russia and National Research University Higher School of Economics, Russia}
\email{SergeevAN@info.sgu.ru}

\author{A.P. Veselov}
\address{Department of Mathematical Sciences,
Loughborough University, Loughborough LE11 3TU, UK  and Moscow State University, Moscow 119899, Russia}
\email{A.P.Veselov@lboro.ac.uk}

\begin{document}

\maketitle

\begin{abstract} We develop the general theory of Jack--Laurent symmetric functions, which are certain generalisations of the Jack symmetric functions, depending on an additional parameter $p_0$. 
\end{abstract}
\tableofcontents

\section{Introduction} 

In the late 1960s Henry Jack \cite{J1, J2} introduced certain symmetric polynomials $Z(\lambda, \alpha)$ depending on a partition $\lambda$ and an additional parameter $\alpha,$ which are now known as {\it Jack polynomials}.
When $\alpha=1$ they reduce to the classical {\it Schur polynomials}, so the Jack polynomials can be considered as a one-parameter generalisation of Schur polynomials, whose theory goes back to Jacobi and Frobenius.  When $\alpha=2$ they are naturally related to zonal spherical functions on the symmetric spaces $U(n)/O(n)$, which was the main initial motivation for Jack. The theory of Jack polynomials was further developed by Stanley \cite{Stanley} and by Macdonald, who also extended them to the symmetric polynomials depending on two parameters, nowadays named after him \cite{Ma}.

Approximately at the same time Calogero \cite{C} and Sutherland \cite{S} initiated the theory of quantum integrable models, describing the interaction particles on the line, which in the classical case were studied by Moser \cite{Moser}. 
Although it was not recognised at the time Jack polynomials can be defined as symmetric polynomial eigenfunctions of (properly gauged) version  ${\mathcal L}_{k, N}$ of the Calogero--Moser--Sutherland (CMS) operator
$$L_k^{(N)} = -\sum_{i=1}^N
\frac{\partial^2}{\partial
z_{i}^2}+\sum_{i<j}^N \frac{2k(k+1)}{\sinh
^2(z_{i}-z_{j})},$$
which in the exponential coordinates 
$x_i = e^{2z_i}$ has the form
\begin{equation}
\label{CM}
 {\mathcal L}_{k, N}=\sum_{i=1}^N
\left(x_{i}\frac{\partial}{\partial
x_{i}}\right)^2-k\sum_{ i < j}^N
\frac{x_{i}+x_{j}}{x_{i}-x_{j}}\left(
x_{i}\frac{\partial}{\partial x_{i}}-
x_{j}\frac{\partial}{\partial
x_{j}}\right)
\end{equation}
where the parameter $k$ is related to Jack's $\alpha$ by $k=-1/\alpha$.

A remarkable property of Jack polynomials is the {\it stability}, which corresponds to the fact that the dependence  ${\mathcal L}_{k, N}$ of on the dimension $N$ can be eliminated by adding a multiple of the momentum (which is an integral of the system): the operators
$$
\tilde {\mathcal L}_{k, N}={\mathcal L}_{k, N}+k(N-1)\sum_{i=1}^N
x_{i}\frac{\partial}{\partial x_{i}} = \sum_{i=1}^N
\left(x_{i}\frac{\partial}{\partial x_{i}}\right)^2-2k\sum_{i=1}^N\left(\sum_{
j \ne i} \frac{x_{i}x_{j}}{x_{i}-x_{j}}\right) \frac{\partial}{\partial
x_{i}}
$$
are stable in the sense that they commute with the natural homomorphisms
$\phi_{M,N}: \Lambda_M \rightarrow \Lambda_N,$ sending $x_i$ with $i>N$ to zero, where 
$\Lambda_{N} = \mathbb C[x_{1},\dots, x_{N}]^{S_N}$ is the algebra of symmetric polynomials.

This allows one to define the Jack symmetric functions $P^{(k)}_{\lambda}$ as elements of $\Lambda$ defined as the inverse limit of $\Lambda_N$ in the category of graded algebras (see \cite{Ma}). The corresponding infinite-dimensional version of the CMS operator has the following explicit form in power sums $p_{a}=x_1^a+x_2^a+\dots, \, a \in \mathbb N$ (see \cite{Stanley, Awata}):
 \begin{equation}\label{infinity}
\tilde{ \mathcal L}_{k}=\sum_{a,b>0}p_{a+b}\partial_{a}\partial_{b}-k\sum_{a,b>0}p_{a}p_b \partial_{a+b}
 +\sum_{a>0}(a+ak-k)p_a\partial_a,
\end{equation}
where $\partial_a = a\frac{\partial}{\partial p_a}.$ Some new explicit formulas for the higher order CMS integrals at infinity were recently found by Nazarov and Sklyanin in \cite{NS1, NS2}.

In the present paper we define and study a Laurent version of Jack symmetric functions - {\it Jack--Laurent symmetric functions} and the corresponding infinite-dimensional Laurent analogue of the CMS operator acting on the algebra $ \Lambda^{\pm}$ freely generated by $p_a$ with $a \in \mathbb Z \setminus \{0\}$ being both positive and negative.
The variable $p_0$ plays a special role and will be considered as an {\it additional parameter}.

The idea to consider the Laurent polynomial eigenfunctions of CMS operator (\ref{CM}) is quite natural and was proposed already by Sutherland in \cite{Suth}. The corresponding Laurent polynomials were later discussed in more details by Sogo \cite{Sogo1, Sogo2, Sogo3}. However, as it was pointed out by Forrester in his MathSciNet review of the paper \cite{Sogo1}, in finite dimension it does not have much sense since the corresponding Laurent polynomials always can be reduced to the usual Jack polynomials simply by multiplication by a suitable power of the determinant $\Delta=x_1\dots x_N.$

In the infinite-dimensional case one can not do this since the infinite product $x_1x_2\dots$ does not belong to $\Lambda.$
Moreover, in the Laurent case {\it there is no stability} (at least in the same sense as above, since one can not set $x_i$ to zero), so the corresponding Jack--Laurent symmetric functions essentially depend on both $k$ and additional parameter $p_0$, which can be viewed as "dimension". Such a parameter appeared already in Jack's paper \cite{J2} as $S_0$ (see page 9 there) and Sogo's papers, but its importance was probably first became clear after the work of Rains \cite{Rains}, who considered $BC$-case (see also \cite{SV3} and \cite{SV5}).

Our main motivation for studying the Jack--Laurent symmetric functions came from the representation theory of Lie superalgebra $\mathfrak{gl}(m,n)$ and related spherical functions, where these functions play an important role. We will discuss this in a separate publication.

The structure of the paper is as follows. In section 2 we introduce the infinite-dimensional Laurent version of the CMS operator 
$$
{ \mathcal L}_{k,p_0}=\sum_{a,b\in \mathbb Z}p_{a+b}\partial_{a}\partial_{b}-k[\sum_{a,b>0}p_{a}p_b \partial_{a+b} -\sum_{a,b<0}p_{a}p_b \partial_{a+b}]
$$
\begin{equation}\label{infinityL}
 - kp_0 [\sum_{a>0} p_{a} \partial_{a}-\sum_{a<0} p_{a} \partial_{a}] +(1+k)\sum_{a\in\mathbb Z}a p_a\partial_a,
\end{equation}
depending on an additional parameter $p_0$, as well as its quantum integrals, acting on $\Lambda^{\pm}.$ Our approach is based on an infinite-dimensional version of Dunkl operator \cite{SV7} and is different from that of \cite{NS1, NS2} (see although the discussion of possible relations in \cite{SV7}). 

 In section 3 we consider the Jack--Laurent polynomials $P^{(k)}_{\chi}(x_1,\dots,x_N)$, which are elements of $\Lambda^{\pm}_{N}=\Bbb C[x_{1}^{\pm1},\dots, x_{N}^{\pm1}]^{S_N}$ parametrized by non-increasing sequences of integers $\chi=(\chi_1, \dots, \chi_N).$ We study their properties, which essentially follow from the usual case. 
 
 In section 4 we define our main object - Jack--Laurent symmetric functions
 $P^{(k,p_0)}_{\alpha}\in \Lambda^{\pm}$, rationally depending on the parameters $k$ and $p_0$ and labelled by bipartitions $\alpha=(\lambda,\mu)$, which are pairs of the usual partitions $\lambda$ and $\mu.$ The defining property is that their images under natural homomorphisms $\varphi_N:  \Lambda^{\pm} \rightarrow \Lambda^{\pm}_{N}$ give the corresponding Jack--Laurent polynomials. 
 An alternative construction of Jack--Laurent symmetric functions, using the monomial symmetric functions, was proposed in \cite{SV5}. 
 We prove the existence of $P^{(k,p_0)}_{\alpha}$ for all $k \notin \mathbb Q$ and $kp_0\neq n+km, \, m,n \in \mathbb Z_{>0}.$ The usual Jack symmetric functions are particular cases corresponding to empty second partition $\mu$:  $P^{(k,p_0)}_{\lambda,\emptyset}=P^{(k)}_{\lambda}\in \Lambda \subset \Lambda^{\pm}.$
 The simplest Laurent example corresponding to two one-box Young diagrams is given by
 $$P^{(k,p_0)}_{1,1}= p_1 p_{-1} - \frac{p_0}{1+k-kp_0}.$$

In sections 5-8 we study the Laurent analogues of Harish-Chandra homomorphism, Pieri and evaluation formulas and compute the square norms of $P^{(k,p_0)}_{\alpha}$ for the corresponding symmetric bilinear form on $\Lambda^{\pm}.$

Section 9 is devoted to an important special case $k=-1,$ corresponding to {\it Schur--Laurent symmetric functions.}
We show that the limit $S_{\lambda,\mu}$ of Jack--Laurent symmetric functions $P^{(k,p_0)}_{\lambda,\mu}$ when $k\rightarrow -1$ for generic $p_0$ does exist, does not depend on  $p_0$ and can be given by an analogue of Jacobi-Trudy formula. The related symmetric Laurent polynomials (called sometimes symmetric Schur polynomials indexed by a composite partition $s_{\bar \mu, \lambda}(x)$) and their supersymmetric versions play an important role in representation theory of Lie superalgebra $\mathfrak{gl}(m,n)$ (see \cite{CK, Moens, CW}).
 
 In the last section we discuss some conjectures and open problems.

\section{Laurent version of CMS operators in infinite dimension}

The finite dimensional CMS operators (\ref{CM}) preserve the algebra of symmetric Laurent polynomials 
$$\Lambda^{\pm}_{N}=\Bbb C[x_{1}^{\pm1},\dots, x_{N}^{\pm1}]^{S_N},$$
generated (not freely) by $p_j(x) = x_1^j + \dots + x_N^j, \, j \in \mathbb Z.$ 

Let us define its infinite-dimensional version - {\it the algebra of Laurent symmetric functions} $\Lambda^{\pm}$ as the commutative algebra with the free generators $p_{i},\, i\in \mathbb Z\setminus \{0\}.$ The dimension $p_0=1+1+\dots +1=N$ does not make sense in infinite-dimensional case, so we will add it as an {\it additional formal parameter,} which will play a very essential role in what will follow.

$\Lambda^{\pm}$ has a natural $\mathbb Z$-grading, where the degree of $p_i$ is $i.$ There is a natural involution $*:\Lambda^{\pm} \rightarrow \Lambda^{\pm}$ defined by 
\begin{equation}\label{inv*}
p_i^* = p_{-i}, \quad i\in \mathbb Z \setminus \{0\}.
\end{equation} 
This algebra can be also represented as
$ \Lambda^{\pm} = \Lambda^+ \otimes \Lambda^{-},$
where $\Lambda^+$ is generated by $p_i$ with positive $i$ and $\Lambda^-$ by $p_i$ with negative $i.$ Note that the involution $*$ swaps $\Lambda^+$ and $\Lambda^-$.

For every natural $N$ there is a homomorphism $\varphi_N:  \Lambda^{\pm} \rightarrow \Lambda^{\pm}_{N}:$
\begin{equation}\label{phin}
\varphi_N (p_j)= x_1^j + \dots + x_N^j, \, j \in \mathbb Z.
\end{equation}
The involution $*$ under this homomorphism goes to the natural involution on $\Lambda^{\pm}_{N}$ mapping $x_i$ to $x_i^{-1}.$

Define also the following algebra homomorphism $\theta: \Lambda^{\pm}\rightarrow \Lambda^{\pm}$ by 
\begin{equation}\label{thet}
\theta(p_a)=k  p_a, \quad a \in \mathbb Z\setminus \{0\}
\end{equation}
(cf. \cite{Ma}, formula (10.6)). If we change also $$k \rightarrow k^{-1}, p_0 \rightarrow kp_0$$ then this map becomes an involution.

Now we are going to construct  explicitly the infinite dimensional version of CMS operator and  higher integrals. Our main tool is an infinite dimensional version of the Dunkl-Heckman operator \cite{H}.

 Let us remind that the {\it Dunkl-Heckman operator} for the root system of the type  $A_n$ has the form
 \begin{equation}
  \label{heckdun}
 D_{i,N}=\partial_i-\frac{k}{2}\sum_{j\ne i}\frac{x_i+x_j}{x_i-x_j}(1-s_{ij}), \quad \partial_i=x_i\frac{\partial}{\partial x_i},
 \end{equation}
 where $s_{ij}$ is a transposition, acting on the functions by permuting the coordinates $x_i$ and $x_j.$
 Heckman proved \cite{H} that the differential operators 
  \begin{equation}
  \label{heckcm}
 \mathcal L^{(r)}_{k,N}=Res \,(D_{1,N}^r+\dots+D^r_{N,N}),
\end{equation}
where $Res$ means the operation of restriction on the space of symmetric polynomials, commute and give 
the integrals for the quantum CMS system with the Hamiltonian $\mathcal H_N= \mathcal L^{(2)}_{k,N}:$
\begin{equation}
\label{CMtrig}
 {\mathcal H}_{N}=\sum_{i=1}^N
\left(x_{i}\frac{\partial}{\partial
x_{i}}\right)^2-k\sum_{ i < j}^N
\frac{x_{i}+x_{j}}{x_{i}-x_{j}}\left(
x_{i}\frac{\partial}{\partial x_{i}}-
x_{j}\frac{\partial}{\partial
x_{j}}\right).
\end{equation}

 We have the following simple, but important Lemma.

 \begin{lemma} The operator  $D_{i,N}$ maps the algebra $\Lambda^{\pm}_{N}[x^{\pm1}_i]$ into itself.
 \end{lemma}
\begin{proof} For the operator $\partial_i=x_i\frac{\partial}{\partial x_i}$ this is obvious since 
$\partial_i(p_l)=lx_i^l$. The operator
 \begin{equation}\label{heckdel}
\Delta_{i,N}=\sum_{j\ne i}\frac{x_i+x_j}{x_i-x_j}(1-s_{ij})
\end{equation}
acts trivially on the algebra $\Lambda^{\pm}_{N}$ and has the property
 \begin{equation}\label{heck*}
\Delta_{i,N}(f^*)=-(\Delta_{i,N}(f))^*.
\end{equation}
Therefore it is enough to prove that  $\Delta_{i,N}(x_i^l)\in \Lambda^{\pm}_{N}[x^{\pm1}_i]$ for $l>0$, which follows from the identity
$$
\Delta_{i,N}(x_i^l)=\sum_{j\ne i}\frac{x_i+x_j}{x_i-x_j}(1-s_{ij})(x_i^l)=\sum_{j\ne i}\frac{x_i+x_j}{x_i-x_j}(x_i^l-x^l_j)
$$
\begin{equation}
 \label{identity}
=x_i^lN+2x_i^{l-1}p_1+\dots+ 2x_ip_{l-1}+p_l-2lx_i^l.
\end{equation}
\end{proof} 

Let $\Lambda^{\pm}[x,x^{-1}]$ be the algebra of Laurent polynomials in  $x$  with coefficients from  $\Lambda^{\pm}$.  Define the differentiation  $\partial$ in  $\Lambda^{\pm}[x,x^{-1}]$ by the formulae
$$
\partial(x)=x, \,\, \partial (p_l)=lx^{l}, 
$$
and the operator $\Delta_{p_0}: \Lambda^{\pm}[x,x^{-1}] \rightarrow \Lambda^{\pm}[x,x^{-1}]$ by
$$
\Delta_{p_0}(x^lf)=\Delta_{p_0}(x^l)f,\,  \Delta_{p_0}(1)=0, \,\,\,f\in \Lambda^{\pm} ,\,l\in \Bbb Z
$$
and
$$
\Delta_{p_0}(x^l)=x^lp_0+2x^{l-1}p_1+\dots+ 2xp_{l-1}+p_l-2lx^l,\,\,  l>0,
$$
$$
 \Delta_{p_0}(x^l)=-(\Delta_{p_0} (x^{-l}))^*, \,\,\,l<0,
$$
where we set $x^*=x^{-1}.$ 
 
Define the {\it infinite dimensional analogue of the Dunkl-Heckman operator} $D_{k,p_0}: \Lambda^{\pm}[x,x^{-1}]\longrightarrow \Lambda^{\pm}[x,x^{-1}]$  by 
\begin{equation}\label{HECK}
D_{k,p_0}=\partial-\frac12k\Delta_{p_0}.
\end{equation}

Let $\varphi_{i,N} : \Lambda^{\pm}[x,x^{-1}]\longrightarrow \Lambda^{\pm}_{N}$ be the homomorphism such that
$$
\varphi_{i,N}(x)=x_i,\,\varphi_{i,N}(p_l)=x_1^l+\dots+x_N^l, \, l \in \mathbb Z
$$ 
and set $p_0=N$. We claim  that  the following diagram
\begin{equation}
\label{comdia9}
\begin{array}{ccc}
\Lambda^{\pm}[x,x^{-1}]&\stackrel{D_{k,p_0}}{\longrightarrow}&\Lambda^{\pm}[x,x^{-1}]\\ \downarrow
\lefteqn{\varphi_{i,N}}& &\downarrow \lefteqn{\varphi_{i,N}}\\
\Lambda^{\pm}_{N}[x_i,x_i^{-1}]&\stackrel{D_{i,N}}{\longrightarrow}& 
\Lambda^{\pm}_{N}[x_i,x_i^{-1}],\\
\end{array}
\end{equation}
 where $D_{i,N}$ are Dunkl-Heckman operators (\ref{heckdun}), is commutative.
This follows from the relations $$
\varphi_{i,N}\circ\partial(x)=\partial_i\circ\varphi_{i,N}(x)=x_i,\,\,\, \varphi_{i,N}\circ\partial(p_l)=\partial_i\circ\varphi_{i,N}(p_l)=l x_i^l,
$$
and
$$
\varphi_{i,N}\circ\Delta_{p_0}(x^lf)=\Delta_{i,N}\circ\varphi_{i,N}(x^lf)
$$
for $p_0=N$ and any $f\in\Lambda^{\pm}$.

 Introduce now a linear operator $E_{p_0} : \Lambda^{\pm}[x,x^{-1}]\longrightarrow \Lambda^{\pm}$  by the formula 
 \begin{equation}\label{E}
  E_{p_0}(x^lf)=p_lf,\, f\in\Lambda,\,\,l\in \Bbb Z
 \end{equation}
and the operators $\mathcal {L}^{(r)}_{k,p_0}: \Lambda^{\pm}\longrightarrow \Lambda^{\pm},\,\, r \in \mathbb Z_+$  by
 \begin{equation}\label{Lr}
 \mathcal{L}^{(r)}_{k,p_0}=E_{p_0}\circ D_{k,p_0}^r,
 \end{equation}
 where the action of the right hand side is restricted to $\Lambda^{\pm}.$

We claim that these operators give a {\it Laurent version of quantum CMS integrals at infinity}.
More precisely, we have the following result.

\begin{thm}\label{Heck} 
The operator $\mathcal{L}^{(r)}_{k,p_0}$ is a differential  operator of order $r$ with polynomial dependence  on $p_0$ and the following properties:
\begin{equation}\label{sym12}
\theta^{-1} \circ \mathcal L^{(r)}_{k,p_0} \circ \theta = k^{r-1} \mathcal L^{(r)}_{k^{-1}, kp_0},
 \end{equation}
\begin{equation}\label{sym13}
(\mathcal L^{(r)}_{k,p_0})^* = (-1)^r\mathcal L^{(r)}_{k,p_0}
\end{equation}
where $\theta$ is defined by (\ref{thet}).
The operator $\mathcal L^{(2)}_{k,p_0}$  is the Laurent version of the CMS operator at infinity given by formula (\ref{infinityL}).

The operators $\mathcal{L}^{(r)}_{k,p_0}$ commute with each other:
$[\mathcal{L}^{(r)}_{k,p_0}, \mathcal{L}^{(s)}_{k,p_0}]=0.$
\end{thm}

\begin{proof} Consider $f\in \Lambda^{\pm}.$ Since $E_{p_0}$ and $\Delta_{p_0}$ commute with multiplication by $f$, we have 
$$
ad(f)^{r+1}(E\circ D_{k,p_0}^r)=E\circ ad(f)^{r+1}(D_{k,p_0}^r)
$$
and therefore
$$
 ad(f)^{r+1}(D_{k,p_0}^r)= ad(f)^{r+1}(\partial^r)=0,
$$
which shows that $\mathcal{L}^{(r)}_{k,p_0}$ is a differential  operator of order $r$.
The formulae (\ref{sym12}), (\ref{sym13}) follow from the symmetries
$$
D_{k,p_0}^*=-D_{k,p_0},\,\, E_{p_0}^*=E_{p_0},\, 
$$
$$
\theta^{-1}E_{p_0}\theta=k^{-1}E_{kp_0},\,\,\,\, \, \theta^{-1}D_{k,p_0}\theta=kD_{k^{-1}, kp_0}
$$
where the action of  $\theta$ is  extended to $\Lambda^{\pm}[x,x^{-1}]$ by $\theta(x)=x.$
The explicit form (\ref{infinityL}) easily follows from a direct calculation.

To prove the commutativity of the integrals note that from (\ref{comdia9}) it follows that
 the  diagram 
 \begin{equation}
 \label{commdi1}
\begin{array}{ccc}
\Lambda^{\pm}&\stackrel{E\circ D_{k, p_0}^r}{\longrightarrow}&\Lambda^{\pm}\\ \downarrow
\lefteqn{\varphi_{N}}& &\downarrow \lefteqn{\varphi_{N}}\\
\Lambda^{\pm}_N&\stackrel{\mathcal L^{(r)}_{k,N}}{\longrightarrow}& 
\Lambda^{\pm}_N, \\
\end{array}
\end{equation}
is commutative, where $\mathcal L^{(r)}_{k,N}$ are the CMS integrals given by Heckman's construction (\ref{heckcm}) and the homomorphism $\varphi_{N}:\Lambda^{\pm} \rightarrow \Lambda^{\pm}_N$ is defined by 
\begin{equation}\label{varphin}
\varphi_N(p_l)=x_1^l+\dots+x_N^l, \, l \in \mathbb Z.
 \end{equation}
 Indeed, for any $f\in \Lambda^{\pm}$ we have 
 $
 D_{k, p_0}^r(f)=\sum_{l}x^lg_l,\,\,\, g_l\in\Lambda^{\pm},
 $
 where the sum is finite. We have
 $
 D_{i,N}^r\circ\varphi_N(f)=\varphi_{i,N}\circ D_{k, p_0}^r(f)=\sum_{l}x_i^l\varphi_N(g_l),
 $
 so
 $$
 \sum_{i=1}^ND_{i,N}^r\circ\varphi_N(f)=\sum_{i=1}^N\sum_{l}x_i^l\varphi_N(g_l)=\sum_{l}\varphi_{N}(p_l)\varphi_N(g_l)=\varphi_N(E(D_{k, p_0}^r(f))),
 $$
 which proves the commutativity of the diagram.
This implies that 
$$
\varphi_N([\mathcal L_{k, p_0}^{(r)},\mathcal L_{k, p_0}^{(s)}](f))=[\mathcal L_{k,N}^{(r)},\mathcal L_{k,N}^{(s)}](\varphi_N(f))=0
$$
since the integrals (\ref{heckcm}) commute \cite{H}. Now the commutativity of the operators $\mathcal{L}^{(r)}_{k,p_0}$
follows from the following useful lemma.

\begin{lemma} 
\label{edin} Let $g$ be an element of $\Lambda^{\pm}$ polynomially depending on $p_0.$ If $\varphi_N(g)=0$ for all $N,$ then $g =0.$
 \end{lemma}

{\it Proof of lemma.} By definition $g$ is a polynomial in a finite number of generators $p_r,\, 0<|r| \le M$ for some $M$ with coefficients polynomially depending on $p_0$. Take $N$ larger than $2M$. Since the corresponding $\varphi_N(p_{r})$ with $0< |r| \le M$ are algebraically independent and $\varphi_N(g)=0$, all the coefficients of $g$ are zero at $p_0=N$. Since this is true for all $N>2M$ the coefficients must be identically zero, and therefore $g=0.$
\end{proof}

\section{Jack--Laurent symmetric polynomials}

As we have already mentioned above the Laurent polynomial eigenfunctions for CMS operators were considered already by Sutherland in \cite{Suth} and later  in more details by Sogo \cite{Sogo1, Sogo2, Sogo3}, who  parametrized these eigenfunctions by the so-called extended Young diagrams, when the negative entries are also allowed. 
Alternatively, one can use two Young diagrams, corresponding to positive and negative parts.  
However, in finite dimension one can always reduce them to the usual Jack polynomials simply by multiplication by a suitable power of the determinant $\Delta=x_1\dots x_N$
(see e.g. Forrester's comment in his MathSciNet review of the paper \cite{Sogo1}).

Let $\chi=(\chi_1, \dots, \chi_N)$ be non-increasing sequence of integers $\chi_1 \ge \chi_2 \ge \dots \ge \chi_N.$
Let $a \in \mathbb Z$ be such that $\nu=\chi +a$ is a partition, which means that $\nu_i=\chi_i+a \ge 0$ for all $i=1,\dots, N.$  
Define the corresponding {\it Jack--Laurent symmetric polynomial} $P^{(k)}_{\chi} \in \Lambda^{\pm}_N$ by 
 \begin{equation}\label{jackchi}
P^{(k)}_{\chi}(x_1,\dots, x_N):=(x_1\dots x_N)^{-a} P^{(k)}_{\nu}(x_1,\dots, x_N),
\end{equation}
where $P^{(k)}_{\nu}(x_1,\dots, x_N)$ are the usual Jack polynomials \cite{Ma}.
It is well-defined because of the well-known property of Jack polynomials
 \begin{equation}\label{jackst}
P^{(k)}_{\nu+b}(x_1,\dots, x_N)=(x_1\dots x_N)^{b} P^{(k)}_{\nu}(x_1,\dots, x_N)
 \end{equation}
 for all $b\ge 0$
 (see e.g. \cite{Stanley}).

There exists a natural involution $*$  on the algebra $\Lambda_N^{\pm}$
$$
x_i^*=x_i^{-1},\,i=1,\dots, N.
$$
The following lemma shows how this involution acts on the Jack--Laurent symmetric polynomials.
 \begin{lemma} For any non-increasing sequence of integers $\chi$
 \begin{equation}\label{w1}
P^{(k)}_{\chi}(x_1,\dots,x_N)^*=P^{(k)}_{w(\chi)}(x_1,\dots,x_N),\,\,\, 
 \end{equation}
 where $w$ is the following involution
  \begin{equation}\label{w2}
w(\chi)=(-\chi_N,\dots,-\chi_1).
 \end{equation}
\end{lemma}
\begin{proof}
It is enough to consider the case when $\chi=\lambda$ is a partition with $ l(\lambda)\le N.$ In that case we have to show that
$$
P^{(k)}_{\lambda}(x_1^{-1},\dots, x^{-1}_n)=(x_1\dots x_N)^{-a}P^{(k)}_{\nu}(x_1,\dots,x_N),
$$
where $\nu=(a-\lambda_N,\dots, a-\lambda_1)$ and $a\ge \lambda_1$. 
Recall that the Jack polynomial $P^{(k)}_{\lambda}(x_1,\dots,x_N)$ can be uniquely  characterised by the following properties: it is an eigenfunction of the CMS operator $\mathcal L_{k, N}$ given by (\ref{CM}) and has an expansion
$$
P^{(k)}_{\lambda}=m_{\lambda}+\sum_{\mu<\lambda}c_{\mu,\lambda}(k)m_{\mu},
$$
where $m_{\mu}$ are the standard monomial polynomials \cite{Ma} and $\mu\le\lambda$ means dominance order: $\mu_1+\dots+\mu_i \le \lambda_1+\dots+\lambda_i,\,\, i=1, \dots, N$ and $ |\lambda|=|\mu|.$
The operator $\mathcal L_{k, N}$ is invariant with respect to involution $x_i\rightarrow x_i^{-1}$, so $P^{(k)}_{\lambda}(x_1^{-1},\dots, x^{-1}_n)$ as well as $(x_1\dots x_N)^aP^{(k)}_{\lambda}(x_1^{-1},\dots, x^{-1}_N)$ are the eigenfunctions of this operator. Since 
$
 (x_1\dots x_N)^am_{\mu}^*=m_{a-\mu_N,\dots,a-\mu_1}
$
we have
$$
(x_1\dots x_N)^aP^{(k)}_{\lambda}(x_1^{-1},\dots, z^{-1}_N)=m_{\nu}+\sum_{\mu<\lambda}c_{\mu,\lambda}m_{a-\mu},
$$
so we only need to show that $a-\mu<a-\lambda$. But the inequalities
$$
 \mu_1+\dots+\mu_i \le \lambda_1+\dots+\lambda_i, \, 1\le i\le n, \quad |\lambda|=|\mu|
$$
imply
$$
\lambda_{N-i+1}+\dots+\lambda_N\le \mu_{N-i+1}+\dots+\mu_N,
$$
and thus
$$
a-\lambda_{N}+\dots+a-\lambda_{N-i+1}\ge a- \mu_{N}+\dots+a-\mu_{N-i+1}.
$$
This proves the lemma. 
\end{proof}  

Now we are going to present the Laurent version of the Harish-Chandra homomorphism.
 Let $\mathcal{D}_N(k)$ be the algebra of quantum integrals of the CMS  generated by the integrals $\mathcal {L}^{(r)}_{k,N}$. The usual {\it Harish-Chandra homomorphism} 
 $$
 \psi_N: \mathcal{D}_N(k) \longrightarrow \Lambda_N(k)
 $$
maps this algebra onto the algebra of {\it shifted symmetric polynomials} $\Lambda_N (k)\subset \mathbb C[t_1, \dots, t_N],$ consisting of polynomials, which are symmetric in shifted variables $t_i+k(i-1), \, i=1,\dots,N.$ It can be defined by
$$ \mathcal{L}P^{(k)}_{\nu}= \psi_N(\mathcal{L})(\nu) P^{(k)}_{\nu}, \, \,  \mathcal{L} \in \mathcal{D}_N(k),$$ where $\nu$ is a partition and $P^{(k)}_{\nu}$ is the usual Jack polynomial.

The $*$-involution on $\Lambda_{N}^{\pm}$ gives rise to the involution on the algebra $\mathcal{D}_N(k),$ which we will be denote by the same symbol: $x_i^*=x_i^{-1}, \,\, \partial_i^*=- \partial_i,$ where as before $ \partial_i=x_i\frac{\partial}{\partial x_i}.$

 \begin{thm}\label{finiteL} 
  For any integral $\mathcal L \in  \mathcal{D}_N(k)$ and any non-increasing sequence of integers $\chi=(\chi_1, \dots, \chi_N)$ we have
  \begin{equation}\label{xi1}
 \mathcal{L} P^{(k)}_{\chi}(x_1,\dots, x_N)=\psi_N(\mathcal{L})(\chi) P^{(k)}_{\chi}(x_1,\dots, x_N)
\end{equation}
and
 \begin{equation}\label{xi2}
\psi_N(\mathcal{L}^*)(\chi)=\psi_N(\mathcal{L})(w(\chi)).
\end{equation}
 \end{thm}
\begin{proof}
 Let us prove first (\ref{xi1}). It is enough to prove this for  the integrals $\mathcal{L}^{(r)}_{k,N}$.  Let  $f_r= \psi_N(\mathcal{L}^{(r)}_{k,N})$.
Since 
$$
D_{i,N}^r (x_1\dots x_N)^aP^{(k)}_{\lambda}(x_1,\dots, x_N)=(x_1\dots x_N)^a(D_{i,N}+a)^r P^{(k)}_{\nu}(x_1,\dots, x_N)
$$
we have 
$$\mathcal{L}^{(r)}_{k,N} ((x_1\dots x_N)^aP_{\lambda})=\left(\sum_{i=0}^{r}{r \choose i}a^{r-i}f_i(\lambda)\right)(x_1\dots x_N)^aP_{\lambda}$$
for all $a.$ For positive $a$ from formula (\ref{jackst}) it follows that 
$$\sum_{i=0}^{r}{r \choose i}a^{r-i}f_i(\lambda)=f_r(  \lambda_1+a,\dots,\lambda_N+a).$$ Since both sides are polynomial this is true for negative $a$ as well, which implies the claim.
 
 The second statement follows from the relations
$$
\psi_N(\mathcal{L}^*)(\chi)P^{(k)}_{\chi}=\mathcal{L}^*(P^{(k)}_{\chi})=(\mathcal{L}(P^{(k)*}_{\chi}))^*=(\mathcal{L}P^{(k)}_{w(\chi)})^*=\psi_N(\mathcal{L})(w(\chi))P^{(k)}_{\chi}.
$$
\end{proof}

Thus we see that the involution $*$ on  the integrals goes  under the Harish-Chandra homomorphism to the involution
$w : \Lambda_{N}(k)\longrightarrow  \Lambda_{N}(k)$ defined by the relation
$$ w(f)(\chi)=f(w(\chi)),\,\,f\in \Lambda_{N}(k).
$$
This involution can be described also in the following way. 
\begin{lemma}  
  Let $p_{r,a,N}\in \Lambda_{N}(k)$ be the shifted power sum defined by
\begin{equation}\label{shiftedp}
p_{r,a,N}(\chi)=\sum_{i=1}^N(\chi_i+k(i-1)+a)^r-\sum_{i=1}^N(k(i-1)+a)^r, \,\, r \in \mathbb N,
\end{equation}
then
\begin{equation}\label{w3}
w(p_{r,a,N})=(-1)^rp_{r,k-kN-a,N}.
\end{equation}
\end{lemma}
\begin{proof} We have 
$$
w(p_{r,a,N})(\chi)=p_{r,a,N}(w(\chi))=
\sum_{i=1}^N(k(i-1)+a-\chi_{N-i+1})^r-\sum_{i=1}^N(k(i-1)+a)^r
$$
$$
=\sum_{i=1}^N(k(N-i)+a-\chi_i)^r-\sum_{i=1}^N(k(N-i)+a)^r
$$
$$
=(-1)^r\sum_{i=1}^N(\chi_i+ki-kN-a)^r-
(-1)^r\sum_{i=1}^N(ki-kN-a)^r=
(-1)^rp_{r,k-kN-a,N)}(\chi).
$$
\end{proof}

Let us present now a Laurent version of Pieri formula.  Define the following functions for positive integers $r,i$ and any $b:$
\begin{equation}\label{cchi}
c_{\chi}(ri,b)=\chi_r-\chi_i-1+k(r+1-i)+b
\end{equation}
and
\begin{equation}\label{V}
V_i(\chi)=\prod_{r=1}^{i-1}\frac{c_{\chi}(ri,1)c_{\chi}(ri,-2k)}{c_{\chi}(ri,1-k)c_{\chi}(ri,-k)}.
\end{equation}
Let $\varepsilon_i$ be the sequence of length $N$ with all zeroes except 1 at $i$-th place.

\begin{thm}\label{JLP} The Jack--Laurent  polynomials satisfy the following Pieri formula:
\begin{equation}\label{pieri}
p_1P^{(k)}_{\chi}(x_1,\dots,x_N)=\sum_{i}V_i(\chi)P^{(k)}_{\chi+\varepsilon_i}(x_1,\dots,x_N),
\end{equation}
where the sum is taken over $1\le i\le N$ such that $\chi+\varepsilon_i$ is a  non increasing sequence  of integers.
\end{thm}
\begin{proof} If $\chi$ is a partition the result is well known \cite{Ma}. In the general case choose integer $a$ such that  $\nu=\chi+a$ is a partition, multiply  both sides of the Pieri formula for $\nu$ by 
$(x_1\dots x_N)^{-a}$ and take into account that $V_i(\chi)=V_i(\chi+a)$.
\end{proof}

We will need also the following corollary of the Pieri formula.

Let $\chi$ be a non-increasing sequence of $N$ integers and set
$$
e_N(\chi)=\sum_{i=1}^N\chi_i^2-k\sum_{i=1}^N(N-2i+1){\chi}_i,
$$
which is the eigenvalue of the CMS operator.
Define the following polynomial in variable $t$ depending on a complex  number $s$
\begin{equation}\label{R}
R_{\chi}(t,s):=\prod_{j}\frac{t-e_N(\chi+\varepsilon_j)}{s-e_N(\chi+\varepsilon_j)},
\end{equation}
where the product is taken over all $j$ such that $\chi +\varepsilon_j$ is a non-increasing sequence of integers and 
$e_N(\chi+\varepsilon_j)\ne s.$

\begin{proposition}\label{JL} Let $k$ be not a positive rational or zero. If $s=e_N(\chi+\varepsilon_i)$ for some $i$ such that  $\chi+\varepsilon_i$ is  a non-increasing sequence of integers, then
$$
R_{\chi}(\mathcal{L}_{k,N},s)(p_1P^{(k)}_{\chi}(x_1,\dots,x_N))=V_i(\chi)P^{(k)}_{\chi+\varepsilon_i}(x_1,\dots,x_N).
$$
 If $s\ne e_N(\chi+\varepsilon_i)$ then 
$$
R_{\chi}(\mathcal{L}_{k,N},s)(p_1P^{(k)}_{\chi}(x_1,\dots,x_N))=0.
$$
\end{proposition}
\begin{proof}
 Since $k$ is not positive rational or zero, then for  $i\ne j$ 
$$
e_N(\chi+\varepsilon_i)-e_N(\chi+\varepsilon_j)=2(\chi_i-\chi_j)+2k(i-j)\ne0.
$$
 In other words all these quantities are pairwise distinct. One can check also that for these $k$  the quantities  $V_i(\chi)$ are well-defined and non-zero. Now the result directly follows from the Pieri formula.
\end{proof}

\section{Jack--Laurent symmetric functions}

Let $\mathcal P$ be the set of all partitions (or, Young diagrams). By {\it bipartition} we will mean any pair of partitions 
$\alpha=(\lambda, \mu) \in \mathcal P \times \mathcal P.$ Define the {\it length of bipartition} $\alpha=(\lambda, \mu)$ by 
$l(\alpha):=l(\lambda)+l(\mu).$ Let $l(\alpha)\le N$ then we set
\begin{equation}\label{chiN}
\chi_N(\alpha)=(\underbrace{\lambda_1,\dots,\lambda_r,0,\dots,0,-\mu_s,\dots,-\mu_1}_{N}).
\end{equation}

Let 
$\varphi_N: \Lambda^{\pm} \rightarrow \Lambda_N^{\pm}$ be the homomorphism defined by 
$\varphi_N(p_i)=x_1^i + \dots + x_N^i, \, i \in \mathbb Z\setminus\{0\}$ and specialization $p_0=N.$

We define now the {\it Jack--Laurent symmetric functions} $P^{(k,p_0)}_{\alpha}\in \Lambda^{\pm}$  
by the following theorem-definition.

\begin{thm} 
\label{th41}
If $k \notin \mathbb Q$, $p_0\neq n+k^{-1}m$ for any $m,n \in \mathbb Z_{>0}$ then for any bipartition $\alpha$ there exists  a unique  element $P^{(k,p_0)}_{\alpha} \in \Lambda^{\pm}$ (called Jack--Laurent symmetric function) such that for every $N \in \mathbb N$
\begin{equation}\label{jlsf}
\varphi_N(P^{(k,p_0)}_{\alpha})=\begin{cases}P^{(k)}_{\chi_N(\alpha)}(x_1,\dots,x_N)\,\textit{if} \,\,l(\alpha)\le N\\
\quad 0\quad\,\textit{if}\,\, l(\alpha)> N.
\end{cases}
\end{equation}
\end{thm}

\begin{proof} Let us prove the existence  first. We will prove it by induction in $|\lambda|$. If 
$|\lambda|=0$ then we set 
$$
P^{(k,p_0)}_{\emptyset,\mu}=P^{(k)*}_{\mu}
$$
where  $P^{(k)}_{\mu}$ is the usual  Jack symmetric  function \cite{Ma}. $P^{(k)}_{\mu}$ does not depend on $p_{0}$ and  for $l(\mu)\le N$ 
$$
\varphi_{N}(P^{(k)*}_{\mu})=(\varphi_{N}(P^{(k)}_{\mu}))^*=P^{(k)}_{\mu}(x_1^{-1},\dots,x_N^{-1})=P^{(k)}_{w(\mu)}(x_1,\dots,x_N).
$$
But $w(\mu)=(-\mu_N-\dots,-\mu_1)=\chi_{N}(0,\mu)$.  If $l(\mu)> N$, then 
$$
\varphi_{N}(P^{(k)*}_{\mu})=(\varphi_{N}(P^{(k)}_{\mu}))^*=0,
$$
so we proved the theorem when $|\lambda|=0$.

Let $\alpha$ be a bipartition. Denote by $X(\alpha)$  and $Y(\alpha)$ the sets of bipartitions, which  can be obtained from $\alpha$ by adding one box to $\lambda$ and by deleting one box from $\mu$ respectively and define
 $$
 Z(\alpha)=X(\alpha)\cup Y(\alpha).
 $$
Similarly to the previous section define for any bipartition $\alpha$
\begin{equation}\label{ep0}
e_{p_0}(\alpha)= \sum \lambda_i^2 + \sum \mu_j^2 + k\sum (2i-1)\lambda_i + k\sum (2j-1)\mu_j - kp_0(|\lambda|+|\mu|)
\end{equation}
and consider the following polynomial in $t,$ depending rationally on $p_0$ and on an additional parameter $s$
$$
R_{\alpha}(p_0,t,s):=\prod_{}\frac{t-e_{p_0}(\gamma)}{s-e_{p_0}(\gamma)},
$$
where the product is over all bipartitions $\gamma\in  Z(\alpha)$ such that $e_{p_0}(\gamma)\ne s.$

 Now suppose that theorem is true for all  $\alpha=(\lambda,\mu)$  with $|\lambda|\le M$. Let $\beta=(\nu,\tau)$ be a bipartition such that  $|\nu|=M+1$. Let $\alpha$ be a bipartition obtained from $\beta$ by removing one box $(i, \nu_i)$ from $\nu$.
Set  $V(\alpha,\beta)=V_i(\lambda)$, where $V_i$ is defined by formula (\ref{V}) with $i$ corresponding to the removed box.
One can check that if $k$ is not rational and $p_0\neq n+k^{-1}m$ with natural $m,n$ the coefficient $V(\alpha,\beta)$ is well defined and nonzero.
Therefore we can define
$$
P^{(k,p_0)}_{\beta}=V(\alpha,\beta)^{-1}R_{\alpha}(p_0,\mathcal L_2,s)(p_1P^{(k,p_0)}_{\alpha})
$$
with $s=e_{p_0}(\beta)$.
Let  $\gamma\in X(\alpha)$ with added box $(j, \lambda_j+1)$, then
$$
e_{p_0}(\beta)-e_{p_0}(\gamma)=2(\lambda_i-\lambda_{j})+2k(i-j).$$
Similarly for $\gamma\in Y(\alpha)$ with deleted box $(j, \mu_j)$
$$
e_{p_0}(\beta)-e_{p_0}(\gamma)=2(\lambda_i+\mu_j)+2k(i+j-1)-2kp_0.
$$
From the previous formula we see that $P_{\beta}$ is well defined if $p_0\neq n+k^{-1}m,\, m,n \in \mathbb Z_{>0}$ and 
$$
\varphi_{N}(P^{(k,p_0)}_{\beta})=V(\alpha,\beta)^{-1}R_{\alpha}(N, \mathcal{L}_{k,N},s)(\varphi_{N}(p_1)\varphi_{N}(P^{(k,p_0)}_{\alpha})).
$$
Now we are going to compare two polynomials  $R_{\alpha}(N, t,s)$ and $R_{\chi_N(\alpha)}(t,s)$.
If $l(\alpha)<N$ then $l(\beta)\le N$ and 
$$
R_{\alpha}(N,t,s)=R_{\chi_N(\alpha)}(t,s),\, V_i(\chi_N(\alpha))=V(\alpha,\beta).
$$
By induction assumption and proposition \ref{JL} we have 
$$
\varphi_{N}(P^{(k,p_0)}_{\beta})=P^{(k)}_{\chi_N(\beta)}(x_1,\dots,x_N).
$$
If $l(\alpha)=l(\beta)=N$ then there exists $\gamma\in X(\alpha)$ with the added box $(l(\lambda)+1,1)$ and it is easy to check that
$$
R_{\alpha}(N,t,s)=R_{\chi(\alpha)}(t,s)\frac{t-e_N(\gamma)}{s-e_N(\gamma)}, \quad
\frac{\mathcal{L}_{k,N}-c_N(\gamma)}{s-c_N(\gamma)}P^{(k)}_{\chi(\beta)}=P^{(k)}_{\chi(\beta)},
$$
which by proposition \ref{JL}  imply  that 
$
\varphi_{N}(P^{(k,p_0)}_{\beta})=P^{(k)}_{\chi(\beta)}(x_1,\dots,x_{N}).
$

Suppose now  that $l(\beta)> N$ and consider two cases: $l(\alpha)>N$ and $l(\alpha)=N$. In the first case by induction $\varphi_{N}(P_{\alpha})=0,$ therefore $\varphi_{N}(P^{(k,p_0)}_{\beta})=0$. In the second case we have again equality 
$
R_{\alpha}(N,t,s)=R_{\chi_N(\alpha)}(t,s),
$
but this time  
$$
s=e_N(\beta)\ne  e_{N}(\chi(\alpha)+\varepsilon_j),\, 1\le j\le N
$$   
and according  to proposition \ref{JL} 
$$
R_{\chi_N(\alpha)}(N,\mathcal{L}_{k,N},s)(\varphi_N(p_1)\varphi_{N}(P^{(k,p_0)}_{\alpha}))=0.
$$
This proves the existence.  The uniqueness follows from the same arguments as in the proof of lemma  \ref{edin}.
\end{proof}

We will show in the next section that the Jack--Laurent symmetric functions $P^{(k,p_0)}_{\alpha}$ are the eigenfunctions of the CMS operators 
$\mathcal{L}^{(r)}_{k,p_0}.$ 

\begin{remark}
The usual definition of Jack symmetric functions uses the basis of monomial symmetric functions. The problem in the Laurent case 
is that the monomial symmetric functions (corresponding to $k=0$) also depend on the additional parameter $p_0$ and the very existence of them (which can be proven in a similar way) is not quite obvious. Having this basis one can define the Jack--Laurent symmetric functions using the CMS operator and show that they have a triangular decomposition in the monomial basis with coefficients polynomially depending on $p_0$.
This implies also that the Jack--Laurent symmetric functions $P^{(k,p_0)}_{\alpha}$ form a basis in $\Lambda^{\pm}$ for all parameters $k, p_0$ satisfying the conditions of Theorem \ref{th41}.
\end{remark}

Here is the explicit form of the Jack--Laurent symmetric functions in the simplest cases:
 \begin{equation}\label{p11}
P^{(k,p_0)}_{1,1}= p_1 p_{-1} - \frac{p_0}{1+k-kp_0},
\end{equation}
 \begin{equation}\label{p111}
P^{(k,p_0)}_{1^2,1}=\frac{1}{2}(p_1^2-p_2)p_{-1} - \frac{2(p_0-1)}{2+4k-2kp_0}p_1,
\end{equation}
where $1^2$ denotes partition $\lambda=(1,1).$ 
%$$
%P^{(k,p_0)}_{2,2}= P_{(2)} P_{(-2)} + \frac{4k}{(k-1)^2}\frac{2-kp_0}{3+k-kp_0}P_1P_{-1}-
%$$
%\begin{equation}\label{p22}
%-\frac{k+1}{(k-1)^2} \frac{2p_0(3-kp_0)}{(3+k-kp_0)(2+k-kp_0)}.
%\end{equation}

Define the involution $w$ on the bipartitions by
$$
w(\alpha)=w(\lambda,\mu):=(\mu,\lambda).
$$

\begin{corollary} For any bipartition $\alpha$ the corresponding Jack--Laurent symmetric functions satisfy the property
$$
P^{(k,p_0)*}_{\alpha}=P^{(k,p_0)}_{w(\alpha)}.
$$
\end{corollary} 
\begin{proof}
Choose $N\ge l(\alpha)$, then we have
$$
\varphi_{N}(P^{(k,p_0)*}_{\alpha})=P^{(k)*}_{\chi_{N}(\alpha)}=P^{(k)}_{w(\chi_{N}(\alpha))}=P^{(k)}_{\chi_{N}(w(\alpha))}=\varphi_{N}(P^{(k,p_0)}_{w(\alpha)}),
$$
which implies the claim.
\end{proof}

\section{Harish-Chandra homomorphism and Polychronakos operator}

%Let $\mathcal {D}(k,p_0)$ be the algebra generated  by the CMS integrals $\mathcal{L}^{(r)}_{k,p_0}.$ 
%The action on the Jack--Laurent symmetric functions $P_{\alpha}$ defines
%the infinite dimensional  Laurent version of the Harish-Chandra homomorphism, which we are going to study now in more detail.

In this section it will be convenient for us to think of $p_0$ as a variable, while $k$ should be still considered as a fixed parameter. The difference between variable and parameter is only in the point of view, which we will continue to change, hopefully without much problems for the reader. 

Recall that the usual Harish-Chandra homomorphism 
$$
\psi_N : \mathcal{D}_N(k)\longrightarrow \Lambda_N(k)
$$
maps the algebra generated by the quantum CMS integrals (\ref{heckcm}) onto the algebra of the shifted symmetric polynomials $\Lambda_N(k).$  The algebra $\mathcal{D}_N(k)$ acts on the algebra of symmetric polynomials $\Lambda_N$ and there is a natural homomorphism 
$$
\phi_{N,N-1} : \mathcal{D}_N(k)\longrightarrow \mathcal{D}_{N-1}(k)
$$ 
induced by the homomorphism $\Lambda_N \rightarrow \Lambda_{N-1}$ sending $x_N$ to zero. 
Consider the inverse limit
$$
\mathcal{D}(k)=\lim_{\leftarrow}\mathcal{D}_N(k),
$$
which we will call the algebra of {\it stable CMS integrals},
and the inverse limit of Harish-Chandra homomorphisms $\psi_N$
$$
\psi : \mathcal{D}(k)\longrightarrow \Lambda(k),
$$
where $\Lambda(k)=\lim_{\leftarrow}\Lambda_N(k)$ is the algebra of the {\it shifted symmetric functions} \cite{OO}.

The algebra $\mathcal{D}(k)$ can be naturally  considered as a subalgebra of the algebra of the differential operators acting on $\Lambda$, which depends on the parameter $k.$ 
It has a natural extension $\mathcal {D}(k)[p_0]=\mathcal {D}(k)\otimes \mathbb C[p_0]$, which also acts on $\Lambda$ if we specialise $p_0.$ 

Let $\mathcal {D}_L(k,p_0)$ be the algebra of differential operators on $\Lambda$ generated over $\mathbb C[p_0]$ by the CMS integrals (\ref{Lr}).  

We claim that any operator from this algebra can be represented as a polynomial in $p_0$ with coefficients, which are stable CMS integrals:
$$
\mathcal {D}_L(k,p_0)=\mathcal {D}(k)[p_0].
$$

To construct the stable CMS integrals we will use the following version of the Dunkl operator, which was introduced by Polychronakos \cite{Poly}:
\begin{equation}\label{poly}
\pi_{i,N}:= x_i\frac{\partial}{\partial x_i}-k\sum_{j\ne i}\frac{x_i}{x_i-x_j}(1-s_{ij})=\partial_i-k\tilde\Delta_{i,N},
\end{equation}
where 
$$
\tilde\Delta_{i,N}:=\sum_{j\ne i}\frac{x_i}{x_i-x_j}(1-s_{ij}).
$$
Note the difference with Dunkl-Heckman operator:
\begin{equation}\label{polydif}
D_{i,N}=\pi_{i,N}+\frac12k\sum_{j\ne i}(1-s_{ij}),
\end{equation}
which implies in particular that $\tilde\Delta_{i,N}$ (and hence $\pi_{i,N}$) do not have symmetry (\ref{heck*}) 
with respect to $*$-involution.

This makes the extension of the operators from $\Lambda$ to $\Lambda^{\pm}$ 
a bit more tricky, so we will consider the Laurent version of the operators. 
The reduction to the usual polynomial case is obvious.

The action of the operator $\pi_{i,N}$ on  $\Lambda_N^{\pm}[x_i^{\pm}]$ is described by
$$
\partial_i(x_i^l)=lx_i^l,\quad \partial_i(p_l)=lx_i^l,
$$
\begin{equation}\label{polyN}
\tilde\Delta_{i,N}(x_i^l)=\begin{cases} (N-l)x_i^l+p_1x_i^{l-1}+\dots+p_{l-1}x_i,\,\,l>0\\
0,\,\,l=0\\ -(lx_i^l+p_{-1}x_i^{l+1}+\dots+p_l),\,l<0.
\end{cases}
\end{equation}
Polychronakos used these operators to give an alternative way to construct the CMS integrals.

\begin{thm} \cite{Poly}
The operators $I^{(r)}_{k,N}=\sum_{i=1}^N\pi^r_{i,N}$
 in $\mathbb C[x_1^{\pm},\dots, x_n^{\pm}]$
commute with each other. Their restrictions to $\Lambda_N^{\pm}$
\begin{equation}\label{polyc}
\mathcal I^{(r)}_{k,N}=Res \sum_{i=1}^N\pi^r_{i,N},\,\, r \in \mathbb Z_{\geq 0}
\end{equation}
 are the commuting quantum integrals  of the CMS system (\ref{CMtrig}). 
\end{thm}

The proof is simple and based on the commutation relations 
\begin{equation}\label{polycom}
[\pi_{i,N}, \pi_{j,N}]=k(\pi_{i,N}-\pi_{j,N})s_{ij}.
\end{equation}

Now we would like to express Heckman's 
CMS integrals $\mathcal L^{(r)}_{k,N}$ via $\mathcal I^{(s)}_{k,N}$ with $ s \leq r.$

\begin{proposition}  
Two series of quantum CMS integrals (\ref{heckcm}),(\ref{polyc}) are related by
\begin{equation}\label{twoin}
\mathcal L^{(r)}_{k,N}= \sum_{a=0}^{r}\mathcal I^{(r-a)}_{k,N}f_a^{(r)}=\sum_{a=0}^{r-1}\mathcal I^{(r-a)}_{k,N}f_a^{(r-1)}, \, r \in \mathbb N,
\end{equation}
where the operator coefficients $f_a^{(r)}$ are defined by the following recurrent relations
\begin{equation}\label{recurr1}
f_a^{(r+1)}=f_a^{(r)}+\frac12 kN f_{a-1}^{(r)},\,\, a\neq r+1,
\end{equation}
\begin{equation}\label{recurr2}
f_{r+1}^{(r+1)}=\frac12 kN f_{r}^{(r)}-\frac12 k \sum_{a=0}^r\mathcal I^{r-a}_{k,N}f_a^{(r)}
\end{equation}
with initial data $f_0^{(0)}=1$ and $f_{a}^{(r)}=0$ when $a<0$ or $a>r.$
\end{proposition}
\begin{proof} 
We claim that 
 \begin{equation}
\label{dinr}
Res \, D_{i,N}^r=Res\, \sum_{a=0}^r\pi_{i,N}^{r-a}f_a^{(r)},
\end{equation}
where $Res$ as before  is the restriction  on symmetric polynomials $\Lambda^{\pm}_{N}.$
Indeed, introduce 
$
S_i=\sum_{j\ne i}(1-s_{ij})
$
and assume (\ref{dinr}).
Then
$$
Res \,D_{i,N}^{r+1}=Res \,\sum_{a=0}^r(\pi_{i,N}+ \frac12kS_i)\pi_{i,N}^{r-a}f_a^{(r)}
=Res \, \sum_{a=0}^r\pi_{i,N}^{r-a+1}f_a^{(r)}$$
$$+ Res \,\frac12k\sum_{a=0}^rS_i\pi_{i,N}^{r-a}f_a^{(r)}
=Res \,\sum_{a=0}^r\pi_{i,N}^{r-a+1}f_a^{(r)}+Res \,\frac12k\sum_{a=0}^r(N\pi_{i,N}^{r-a}-I^{(r-a)}_{k,N})f_a^{(r)}
$$
$$
=Res \,\sum_{a=0}^r\pi_{i,N}^{r-a+1}(f_a^{(r)}+ \,\frac12kNf^{(r)}_{a-1})+Res \,\frac12k(Nf^{(r)}_r-\sum_{a=0}^rI^{(r-a)}_{k,N}f_a^{(r)}),
$$
which leads to the recurrent relations (\ref{recurr1}),(\ref{recurr2}).

Now using this we have from (\ref{heckcm}) for $r \geq 1$
$$
\mathcal L^{(r)}_{k,N}=\sum_{a=0}^r\mathcal I^{(r-a)}_{k,N} f_a^{(r)}=\sum_{a=0}^{r-1}\mathcal I^{(r-a)}_{k,N}(f_a^{(r-1)}+\frac12kNf^{(r-1)}_{a-1})
$$
$$+\frac12 kN^2 f_{r-1}^{(r-1)}-\frac12 kN \sum_{a=0}^{r-1}\mathcal I^{(r-a-1)}_{k,N}f_a^{(r-1)}
=\sum_{a=0}^{r-1}\mathcal I^{(r-a)}_{k,N}f_a^{(r-1)}
$$
since $\mathcal I^{(0)}_{k,N}=N.$ Note also that for $r=1,2$ we have $\mathcal L^{(r)}_{k,N}=\mathcal I^{(r)}_{k,N}.$
\end{proof}

 Let us define now the {\it infinite dimensional analogue of the Polychronakos operator} on $\Lambda^{\pm}[x]$ by the formulas 
 $$
 \pi_{k,p_0}=\partial-k\tilde\Delta_{p_0},\,\,\partial(x^l)=lx^l,\,\,\partial(p_l)=lx^l,\,\,l\in \Bbb Z,
 $$
 $$
 \tilde\Delta_{p_0}(x^l)=\begin{cases} x^l(p_0-l)+x^{l-1}p_1+\dots+xp_{l-1},\,\,l>0\\
0,\,\,l=0\\ -(lx^l+x^{l+1}p_{-1}+\dots+p_l),\,l<0.
\end{cases}
$$

 Let $E_{p_0}: \Lambda^{\pm}[x,x^{-1}]\longrightarrow \Lambda^{\pm}$  be the operator defined by formula (\ref{E}) above
and consider the following set of infinite dimensional CMS integrals 
 \begin{equation}\label{pir}
 \mathcal{I}^{(r)}_{k,p_0}=E_{p_0}\circ \pi_{k,p_0}^r, \,\, r \in \mathbb Z_{\geq 0},
 \end{equation}
 where again the action of the right hand side is restricted to $\Lambda^{\pm}.$
Their commutativity follows from the same arguments as in Theorem \ref{Heck}.
%Note that these integrals polynomially depend on the parameters $k$ and $p_0.$

Let $\mathcal {D}_I(k,p_0)$ be the algebra of differential operators on $\Lambda$ generated over $\mathbb C[p_0]$ by the CMS integrals $\mathcal{I}^{(j)}_{k,p_0}, j \in \mathbb N$.

\begin{proposition}    
$$
\mathcal D_I(k,p_0)=\mathcal D_L(k,p_0).
$$ 
\end{proposition}
\begin{proof}
It is enough to show that  the integrals $\mathcal{L}^{(j)}_{k,p_0}$ can be expressed polynomially through integrals $\mathcal{I}^{(i)}_{k,p_0}, \, i\leq j$  with 
coefficients polynomially depending on $p_0,$ and vice versa.

Define the operators $\hat f_a^{(r)}$ recursively by
\begin{equation}\label{recurrr1}
\hat f_a^{(r+1)}=\hat f_a^{(r)}+\frac12 kp_0 \hat f_{a-1}^{(r)},\,\, a=0, \dots, r,
\end{equation}
\begin{equation}\label{recurrr2}
\hat f_{r+1}^{(r+1)}=\frac12 kp_0 \hat f_{r}^{(r)}-\frac12 k \sum_{a=0}^r\mathcal I^{(r-a)}_{k,p_0}\hat f_a^{(r)}
\end{equation}
with initial data $\hat f_0^{(0)}=1$ and $\hat f_{a}^{(r)}=0$ when $a<0$ or $a>r.$
One can check that for such operators we have the relation
$$
D_{k,p_0}^r= \sum_{a=0}^r \pi_{k,p_0}^r \hat f_a^{(r)}
$$
and hence the relation between the corresponding CMS integrals
\begin{equation}\label{twoint}
\mathcal L^{(r)}_{k,p_0}= \sum_{a=0}^{r}\mathcal I^{(r-a)}_{k,p_0}\hat f_a^{(r)}=\sum_{a=0}^{r-1}\mathcal I^{(r-a)}_{k,p_0}\hat f_a^{(r-1)}, \, r \in \mathbb N.
\end{equation}
Since $\hat f_0^{(r)}=1$ we can reverse these formulas to express polynomially 
$ \mathcal I^{(r)}_{k,p_0}$ via $\mathcal L^{(a)}_{k,p_0}$ with $a\leq r.$ In particular, for $r=1,2$ we have $\mathcal I^{(r)}_{k,p_0}= \mathcal L^{(r)}_{k,p_0}.$
\end{proof}

Now we define the third set of CMS integrals $\mathcal{H}_{k}^{(r)}$ (see formula (\ref{stabint}) below), 
which do not depend on $p_0$ and generate the algebra $\mathcal D(k).$
It turns out that all these three sets of CMS integrals generate the same algebra 
if we allow the coefficients to depend polynomially on $p_0.$

\begin{thm} 
\begin{equation}
\label{thm55}
\mathcal D(k)[p_0]=\mathcal D_L(k,p_0).
\end{equation}
\end{thm}

\begin{proof}  To define the stable quantum CMS integrals we note first that on the algebra $\Lambda[x]$ we have
\begin{equation}\label{stabi}
(\pi_{k,p_0}+kp_0)^{r-1}\partial =\lim_{\longleftarrow}(\pi_{i,N}+kN)^{r-1}\partial_{i}, \, \, r \in \mathbb N. 
\end{equation}
Indeed, the operator $\partial_{i}$ maps algebra  $\Lambda_{N}[x_i]$ into the ideal  $J$ generated by $x_i$, while the operator $\pi_{i,N}+kN$ maps the ideal  $J$ into itself (see formula (\ref{polyN})). Therefore the operator  $(\pi_{i,N}+kN)^{r-1}\partial_{i}$  maps the algebra $\Lambda_{N}[x_i]$ into the ideal $J$ and can be checked to be stable. 

Consider the operators 
$$
H^{(r)}_{i,N}:=\left(\pi_{i,N}+kN\right)^{r-1}\pi_{i,N}:\Lambda_{N}\longrightarrow \Lambda_{N}[x_i].
$$
Note that on $\Lambda_N$ we have $\pi_{i,N}=\partial_{i}.$ From (\ref{stabi}) it follows that these operators are stable and their inverse limit can be naturally identified with 
$$
 H_{k}^{(r)}:=\left(\pi_{k,p_0}+kp_0\right)^{r-1}\pi_{k,p_0} :\Lambda\longrightarrow \Lambda[x].
$$
Thus the integrals 
\begin{equation}
\label{stabint}
\mathcal{H}_{k}^{(r)}=E_{p_0} \circ H_{k}^{(r)}=\sum_{j=1}^{r} {r-1\choose j-1} (kp_0)^{r-j} \mathcal I^{(j)}_{k,p_0}, \, r \in \mathbb N
\end{equation}
do not depend on $p_0$ and belong to $\mathcal D(k)$.  
In particular, for $r=2$ we have $\mathcal{H}_{k}^{(2)}=\mathcal{I}_{k,p_0}^{(2)}+kp_0\mathcal{I}_{k,p_0}^{(1)}=\mathcal{L}_{k,p_0}^{(2)}+kp_0\mathcal{L}_{k,p_0}^{(1)}$, which is the stable version of the CMS operator (\ref{infinityL}) on $\Lambda$ (see \cite{Awata,Stanley}):
\begin{equation}\label{infinityH}
{ \mathcal H}_{k}^{(2)}=\sum_{a,b\in \mathbb N}p_{a+b}\partial_{a}\partial_{b}-k\sum_{a,b \in \mathbb N}p_{a}p_b \partial_{a+b}
 +(1+k)\sum_{a\in\mathbb N}a p_a\partial_a.
\end{equation}
From (\ref{stabint}) one can also express $\mathcal{I}_{k,p_0}^{(r)}$ as a polynomial of $\mathcal{H}_{k}^{(s)}$ with coefficients polynomially depending on $p_0.$
Since 
$\mathcal{H}_{k}^{(r)}, \, r \in \mathbb N$ generate the algebra $\mathcal D(k)$ 
the theorem now follows from the previous proposition.
\end{proof}

This theorem allows us, in particular, to define the Harish-Chandra homomorphism
$
\psi: \mathcal D_L(k,p_0) \rightarrow \Lambda(k)[p_0].
$

Let
$$
p_{r,a}=\lim_\leftarrow p_{r,a,N} \in \Lambda(k)
$$
be the inverse limit of shifted power sums (\ref{shiftedp}). They generate the algebra $\Lambda(k)[p_0]$ over $\mathbb C[p_0].$
\begin{thm} 
The Harish-Chandra homomorphism
$$
\psi : \mathcal {D}(k)[p_0]\longrightarrow \Lambda(k)[p_0]
$$
of algebras over $\mathbb C[p_0]$ transforms  the involution $*$ on the algebra $\mathcal {D}(k)[p_0]$ to the involution 
$w$ on $\Lambda(k)[p_0]$ defined by 
\begin{equation}\label{w12}
w(p_{r,a})=(-1)^r p_{r,k-kp_0-a}.
\end{equation}
More precisely, the following diagram is commutative:
$$
\begin{array}{ccc}
\mathcal {D}(k)[p_0]&\stackrel{*}{\longrightarrow}&\mathcal {D}(k)[p_0]\\ \downarrow
\lefteqn{\psi}& &\downarrow \lefteqn{\psi}\\
\Lambda(k)[p_0]&\stackrel{w}{\longrightarrow}& 
\Lambda(k)[p_0]. \\
\end{array}
$$
\end{thm}
\begin{proof} 
Let $\varphi_N$ be defined by (\ref{varphin}) and use the same notation for its action on  $\mathcal D(k)[p_0]$. Define also  $\phi_N: \Lambda(k)[p_0] \rightarrow \Lambda_N(k)$ by setting $p_0=N$ and 
$$\phi_N(p_{r,0})=p_{r,0,N}.$$
We have the following commutative diagram 
$$
\begin{array}{ccc}
\mathcal D(k)[p_0]&\stackrel{\psi}{\longrightarrow}& \Lambda(k)[p_0]\\
\downarrow \lefteqn{\varphi_N}& &\downarrow \lefteqn{\phi_N}\\
\mathcal D_N(k)&\stackrel{\psi_N}{\longrightarrow}&\Lambda_{N}(k),\\
\end{array} 
$$
so
$$
\phi_N(\psi (D^*)=\psi_N(\varphi_N(D^*))=w_N(\psi_N(\varphi_N(D))).
$$
Therefore we only need to prove that 
$
\phi_N(w(p))=w_N(\phi_N(p))
$
for any $p\in \Lambda(k)[p_0]$. It is enough to show this for shifted power sums, when
this reduces to the identity
$$
p_{r,a,N}(w(\chi))=(-1)^r p_{r,k-kN-a,N}(\chi),
$$
which is easy to check.
Theorem is proved.
\end{proof}

We can consider the elements from $\Lambda(k)[p_0]$  as  functions on bipartitions. Namely, for any bipartition $\alpha=(\lambda,\mu)$ define a homomorphism
$f_{\alpha} :\Lambda(k)[p_0]\longrightarrow \mathbb C[p_0]$ by
 \begin{equation}\label{fa}
f_{\alpha}(p_r)=p_r(\lambda)+w(p_r)(\mu),
\end{equation}
where $p_r=p_{r,0},$ and define for any $p\in \Lambda(k)[p_0]$ the value $p(\alpha)$ by $$p(\alpha):=f_{\alpha}(p).$$
 
 \begin{corollary}  Jack--Laurent symmetric functions are the eigenfunctions of the algebra $\mathcal D(k)[p_0]$ and for any $D\in \mathcal D(k)[p_0]$ we have
 \begin{equation}\label{HCh}
 DP^{(k,p_0)}_{\alpha}=\psi(D)(\alpha)P^{(k,p_0)}_{\alpha}.
 \end{equation}
 \end{corollary}
 \begin{proof} Apply to both sides of (\ref{HCh}) the homomorphism $\varphi_{N}$ with $N\ge l(\lambda)+l(\mu)$. From (\ref{thm55}) and (\ref{commdi1}) we have  
 $$
 \varphi_{N} (DP^{(k,p_0)}_{\alpha})= \varphi_{N} (D) \varphi_N (P^{(k,p_0)}_{\alpha})=  \varphi_{N} (D) P^{(k)}_{\chi_N(\alpha)} 
$$
 $$
 =\psi_N(\varphi_N(D))(\chi_N(\alpha))P^{(k)}_{\chi_N(\alpha)}
=\phi_N(\psi(D))(\chi_N(\alpha))P^{(k)}_{\chi_N(\alpha)}.
 $$
 On the other hand we have
 $$
 \varphi_N(\psi(D)(\alpha)P^{(k,p_0)}_{\alpha})=\varphi_N(\psi(D)(\alpha))\varphi_N(P^{(k,p_0)}_{\alpha})=\varphi_N(\psi(D)(\alpha))P^{(k)}_{\chi_N(\alpha)},
 $$
 so we only need to prove that
 $
 \varphi_{N}(\psi(D)(\alpha))=\phi_N(\psi(D))(\chi_{N}(\alpha)).
 $
 It is enough to check this for $\psi(D)=p_r$.  We have
  $$
  \varphi_{N}( p_r(\alpha))=\varphi_{N}(p_r(\lambda))+\varphi_{N}(w(p_r)(\mu))
 = \varphi_{N}(p_r(\lambda))+w_N(\varphi_{N}(p_r(\mu))$$
 $$=p_{r,N}(\chi_N(\lambda))+p_{r,N}(w_N(\chi_N(\mu))).
  $$
 Since $N\ge l(\lambda)+l(\mu)$ we have
 $$\chi_N(\alpha)=\chi_N(\lambda)+w_N(\chi_N(\mu)).$$
It is easy to see that since for any $1\le i\le N$,  $\lambda_iw_N(\mu)_i=0$ we have
  $$
  p_{r,N}(\chi_N(\lambda))+ p_{r,N}(w_N(\chi_N(\mu)))=p_{r,N}(\chi_{N}(\alpha))=\phi_N(p_r)(\chi_{N}(\alpha)),
  $$
  which completes the proof.
 \end{proof}
 
 Now we can use this to prove the following important result. Let us consider again $p_0$ as a parameter.
 
 \begin{thm}
\label{simpl}
If $k \notin \mathbb Q$ and $kp_0\neq m+nk$ for all $m,n \in \mathbb Z_{>0}$ then the spectrum of the algebra $\mathcal {D}(k)[p_0]$ of quantum CMS integrals on $\Lambda^{\pm}$  is simple.
\end{thm}
\begin{proof} 
Consider the following  shifted symmetric functions
\begin{equation}
\label{bl}
b_{l}(k,a)(x)=\sum_{i\ge1}^{\infty}\left[ B_{l}(x_i+k(i-1)+a)-B_{l}(k(i-1)+a)\right],
\end{equation}
where $B_l(z)$ are the classical Bernoulli polynomials \cite{WW} and $a$ is a parameter. 
They generate the algebra of shifted symmetric functions $\Lambda(k)$.

\begin{lemma}\label{form} We have the following formula
$$
b_{l}(k,a)(\alpha)=l\sum_{\Box\in \lambda} c(\Box,0)^{l-1}
+(-1)^l l\sum_{\Box\in \mu} c(\Box,1+k-kp_0)^{l-1},
$$
where for $\Box=(ij)$ we define
$$
c(\Box,a):=(j-1)+k(i-1)+a.
$$
\end{lemma}
\begin{proof} We have $$B_l(z)=\sum_{s=0}^l b_{ls} z^s, \,\, b_{ls}={l\choose s} B_{l-s},$$ 
where $B_j$ are the Bernoulli numbers.
This implies  
$
b_{l}(k,a)=\sum_{s}b_{ls}p_{s,a}
$
and thus by (\ref{w12})
$$
w(b_{l}(k,a))=\sum_{s}b_{ls}(-1)^sp_{s,k-kp_0-a}.
$$
Using the standard property of Bernoulli polynomials 
$
 B_{l}(-x)=(-1)^lB_{l}(1+x)
$
(see \cite{WW}), we have
$$
\sum_{s}b_{ls}(-1)^sp_{s,k-kp_0-a}=\sum_{s}b_{ls}p_{s,1+k-kp_0-a}.
$$
Therefore 
$
w(b_{l}(k,a))=(-1)^lb_{k,1+k-kp_0}
$
and now lemma follows from the equality 
\begin{equation}
\label{Bernoul1}
b_{l}(\lambda,k,a)=l\sum_{(i,j)\in \lambda} \left[(j-1)+k(i-1)+a\right]^{l-1}.
\end{equation}
\end{proof}

Let us assume now that $b_{l+1}(\alpha)=b_{l+1}(\tilde\alpha).$ Then we have
$$
\sum_{x\in\lambda} c(x,0)^l+(-1)^{l+1}\sum_{y\in\tilde\mu} c(y,1+k-kp_0)^l=\sum_{x\in\tilde\lambda} c(x,0)^l+(-1)^{l+1}\sum_{y\in\mu} c(y,1+k-kp_0)^l.
$$
If this is true for all $l \in \mathbb Z_{\ge 0},$ then the sequences 
$$
(c(x,0), -c(y,1+k-kp_0))_{x\in\lambda,y\in\tilde\mu},\quad (c(x,0), -c(y,1+k-kp_0))_{x\in\tilde\lambda,y\in\mu}
$$
coincide  up to a permutation. Therefore we have for every $x\in \lambda$ two possibilities: $c(x,0)=c(\tilde x,0)$ for some $\tilde x\in\tilde\lambda$, or $c(x,0)=-c(\tilde y,1+k-kp_0)$ for some $\tilde y\in\mu$. In the first case we have for
 $x=(ij),\,\tilde x=(\tilde i\tilde j),$ so that $j-\tilde j+k(i-\tilde i)=0$,
so $j=\tilde j,\,i=\tilde i$ since $k$ is not rational. 

In the second case we have for $
 \tilde y=(\tilde i\tilde j)$ that $kp_0=j+\tilde j-1+k(i+\tilde i-1),$
which contradicts to our assumption, since both $j+\tilde j-1$ and $i+\tilde i-1$ are positive integers.
\end{proof}

\begin{corollary}\label{cor22}
Jack--Laurent symmetric functions obey the following $\theta$-duality property 
\begin{equation}
\label{thetadua}
\theta^{-1}(P^{(k,p_0)}_{\alpha})=d_{\alpha}P^{(k^{-1},kp_0)}_{\alpha'}
\end{equation}
with some constants $d_{\alpha}=d_{\alpha}(k, p_0)$ and $\alpha'=(\lambda',\mu')$.
\end{corollary}
\begin{proof}
Indeed, because of the symmetry property (\ref{sym12}) of quantum integrals $\mathcal{L}_{k,p_0}^{(q)}$ the function $\theta^{-1}(P^{(k,p_0)}_{\alpha})$ is also an eigenfunction of the operator $\mathcal{L}_{k^{-1},kp_0}^{(q)}$ with the same eigenvalue up to a constant. 
 Now the claim follows from the lemma (\ref{form}), the duality property 
 $$
b_{l}(\lambda,k,a)=k^{l-1}b_{l}(\lambda',k^{-1},k^{-1}a)
 $$
  and the simplicity of the spectrum for generic $k$ and $p_0$. An explicit form of the constants $d_{\alpha}(k, p_0)$ is given below by (\ref{duad}).
\end{proof}

\section{Pieri formula for Jack--Laurent symmetric functions}

Let $\alpha$ be a bipartition represented by a pair of Young diagrams $\lambda$ and $\mu$. 
Define for any positive integers $i,j$ the following functions
$$
c_{\lambda}(ji,a)=\lambda_{i}-j-k(\lambda^{\prime}_{j}-i)+a,\;\;
$$
$$
c_{\alpha}(ji,a)=\lambda_{i}+j+k(\mu^{\prime}_{j}+i)+a,
$$
where $\lambda^{\prime}$ as before is the Young diagram conjugated (transposed) to $\lambda$.

Let $x=(ij)$ be a box such that the union $\lambda+x:= \lambda \cup x \in \mathcal P$ is also a Young diagram and similarly to the Pieri formula for Jack polynomials \cite{Ma} define
\begin{equation}\label{V1}
V(x, \alpha)=\prod_{r=1}^{i-1}\frac{c_{\lambda}(jr,1)c_{\lambda}(jr,-2k)}{c_{\lambda}(jr,-k)c_{\lambda}(jr,1-k)}.
\end{equation}
If $x$ can not be added to $\lambda$ we assume that the corresponding $V(x,\alpha)=0.$

Similarly, if the box $y=(ij)$ can be removed from the Young diagram $\mu$ in the sense that $\mu-y:=\mu \setminus y \in \mathcal P$ we define
$$
U(y, \alpha)=\prod_{r=i+1}^{l(\mu)}\frac{c_{\mu}(jr,1+k)c_{\mu}(jr,-k)}{c_{\mu}(jr,1)c_{\mu}(jr,0)}
$$
$$\times\prod_{r=1}^{l(\lambda)}\frac{c_{\alpha}(jr,-1-k(p_{0}+2))c_{\alpha}(jr,-kp_{0})}{c_{\alpha}(jr,-1-k(p_{0}+1))c_{\alpha}(jr,-k(p_{0}+1))}
$$
\begin{equation}\label{V2}
\times\frac{(j-1+k(l(\lambda)+\mu^{\prime}_{j}-p_{0}-1))(j+k(\mu^{\prime}_{j}-l(\mu)))}
{(j+k(l(\lambda)+\mu^{\prime}_{j}-p_{0}))(j-1+k(\mu^{\prime}_{j}-l(\mu)-1))},
\end{equation}
where $l(\lambda)$ is the length, which is the number of non-zero parts in partition $\lambda.$
If $y=(ij)$ can not be removed from $\mu$ we define $U(y, \alpha)=0.$

The following theorem follows from the Pieri formula for Jack--Laurent polynomials (\ref{pieri}).

\begin{thm}
The Jack--Laurent symmetric functions $P_{\lambda,\mu}=P^{(k,p_0)}_{\alpha}$ with $\alpha=(\lambda,\mu)$ satisfy the following Pieri formula:
\begin{equation}\label{pieri1}
p_{1}P_{\lambda,\mu}=\sum_{x}V(x, \alpha)P_{\lambda+x,\mu}+\sum_{y}U(y,\alpha)P_{\lambda,\mu - y}.
\end{equation}
\end{thm}

One can rewrite the formula in terms of the following diagrammatic representation of a bipartition $\alpha=(\lambda,\mu)$
(cf. \cite{MVDJ}).
Consider the following geometric figure $Y=Y_{\lambda,\mu}=Y_{\lambda}\cup Y_{-\mu}\cup \Pi_{\lambda,\mu},$
where
$$
Y_{\lambda}=\{(ji)\mid\; j,i\in\Bbb Z,\;1\le i\le l(\lambda),\, 1\le j\le \lambda_{i}\},
$$
$$
Y_{-\mu}=\{(ji)\mid\; j,i\in\Bbb Z,\;-l(\mu)\le i\le -1,\, -\mu_{i}\le j\le -1\}
$$
and
$$
\Pi_{\lambda,\mu}=\{(ji)\mid\; j,i\in\Bbb Z,\;1\le i\le l(\lambda),\, - l(\mu^{\prime})\le j\le -1\}.
$$
On Fig. 1 we have the corresponding representation for $\lambda=(6,5,4,2,1)$ and $\mu=(7,3,2,1,1).$
Note that for $\lambda$ we follow the French way of drawing Young diagram, for $\mu$ it is rotated by 180 degrees.

\begin{figure}
\centerline{ \includegraphics[width=18cm]{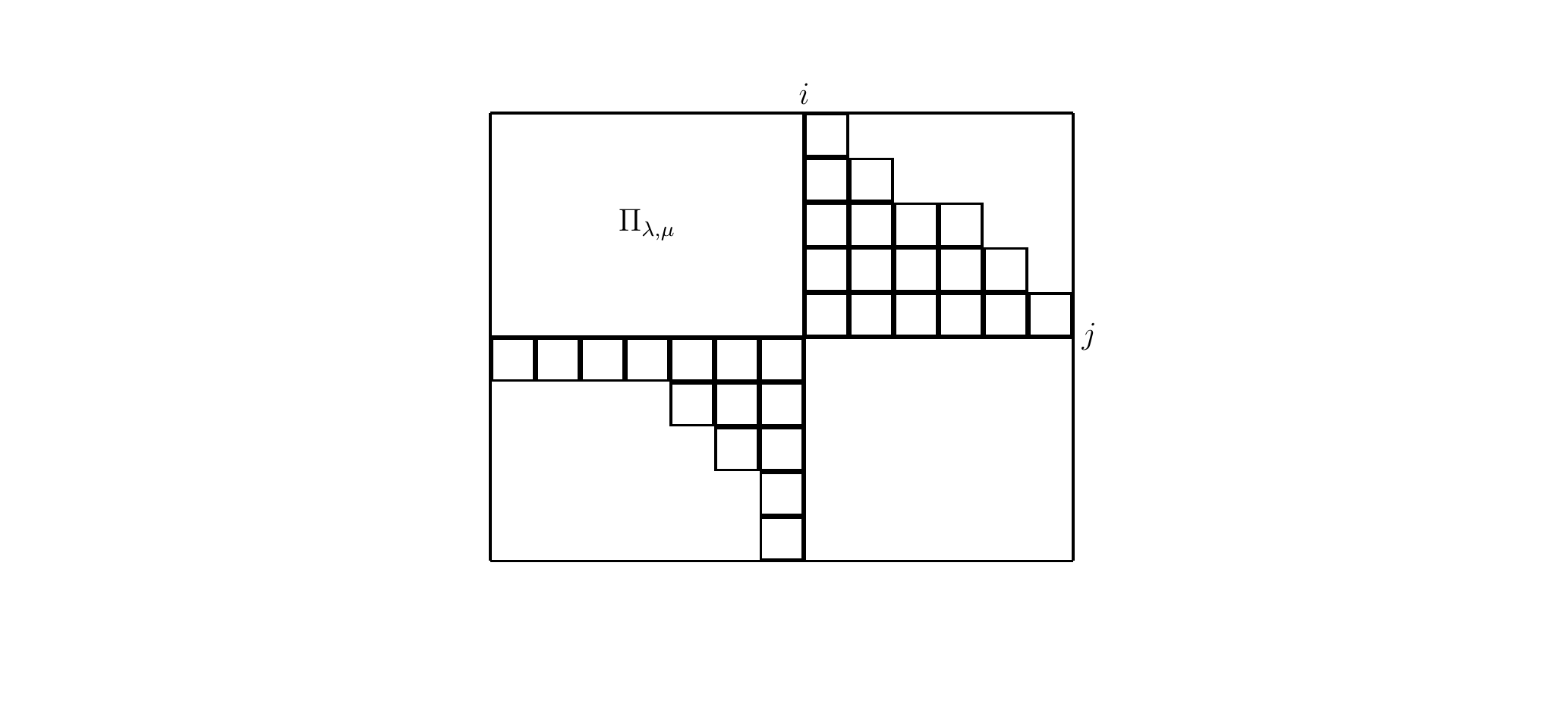}}
\caption{Diagrammatic representation of a pair of partitions} \label{pair}
\end{figure}

Define the following analogues of rows
$$
y_{i}=\begin{cases}\lambda_{i},& 1\le i\le l(\lambda)\cr -\mu_{-i},& -l(\mu)\le i \le -1 \end{cases}
$$
and columns
$$
y^{\prime}_{j}=\begin{cases}\lambda^{\prime}_{j},& 1\le j\le l(\lambda^{\prime})\cr -\mu^{\prime}_{-j},& -l(\mu^{\prime})\le j \le -1 \end{cases}
$$
with all other $y_{i}$, $y^{\prime}_{j}$ being zero.
For every box $\Box$ with integer coordinates $(j,i)$ define the function 
$$
c_{Y}(\Box,a)=y_{i}-j-k(y^{\prime}_{j}-i)+a.
$$
Define for the added box $\Box=(j,i)$ the following subset in $Y_{\lambda}$
$$
\pi_{1}=\{(j,r)\mid 1\le r< i\}
$$ 
and  for  deleted  box  $\Box=(ji)$  the subsets in $Y$
$$
\pi_{2}=\{(j,r)\mid -l(\mu)  \le r<-\mu^{\prime}_{-j} \},
$$
$$
\pi_{3}=\{(j,r)\mid 1 \le r\le l(\lambda)\}.
$$
The meaning of these subsets is clear from Fig. 2, where the deleted box is black and the added box is crosshatched.

\begin{figure}
\centerline{ \includegraphics[width=18cm]{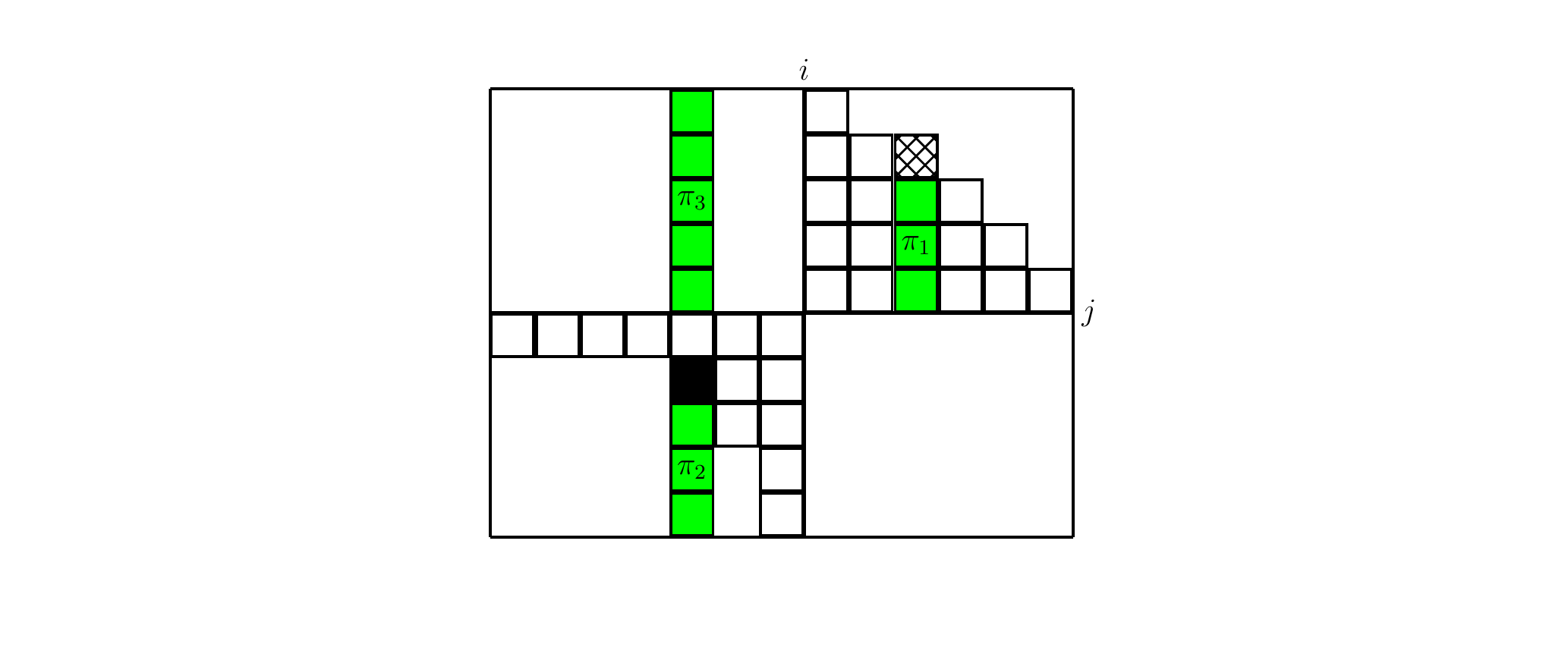}}
\caption{Summation sets for the Pieri formula} \label{Pieri}
\end{figure}

In these terms the Pieri formula (\ref{pieri1}) can be written as
\begin{equation}\label{pieri2}
p_{1}P_{\lambda,\mu}=\sum_{\Box}V(\Box, \alpha)P_{\lambda\cup\Box,\mu}+\sum_{\Box}U(\Box,\alpha)P_{\lambda,\mu \setminus \Box}
\end{equation}
with
\begin{equation}\label{V11}
V(\Box,\alpha)=  \prod_{\Box\in\pi_1}\frac{c_{Y}(\Box,-2k)c_{Y}(\Box,1)}{c_{Y}(\Box,-k)c_{Y}(\Box,1-k)},
\end{equation}
$$
U(\Box,\alpha)=  \prod_{\Box\in\pi_2}\frac{c_{Y}(\Box,-1-k)c_{Y}(\Box,k)}{c_{Y}(\Box,-1)c_{Y}(\Box,0)}
$$
$$
\times\prod_{\Box\in\pi_3}\frac{c_{Y}(\Box,-1-k(p_{0}+2))c_{Y}(\Box,-kp_{0})}{c_{Y}(\Box,-1-k(p_{0}+1))c_{Y}(\Box,-k(p_{0}+1))}
$$
\begin{equation}\label{V21}
\times\frac{(j+1+k(y^{\prime}_{j}-l(\lambda)+p_{0}+1))(j+k(y^{\prime}_{j}+l(\mu))}
{(j+k(y^{\prime}_{j}-l(\lambda)+p_{0}))(j+1+k(y^{\prime}_{j}+l(\mu)+1)}
\end{equation}
with the convention that the product over empty set is equal to 1.

A non-symmetry between $\lambda$ and $\mu$ is due to the choice of $p_1$ in the left hand side of the Pieri formula (\ref{pieri1}). By applying $*$-involution to formula (\ref{pieri1}) one has the corresponding formula for $p_{-1}$, where the roles of $\lambda$ and $\mu$ are interchanged. 

Another remark is that in the Pieri formula (\ref{pieri2}) one can replace the rectangle containing the figure $Y$ by any bigger rectangle with $-M\leq i \leq L$ by changing in formula (\ref{V21}) the lengths $l(\lambda)$ and $l(\mu)$ to $L$ and $M$ respectively.

\section{Evaluation  theorem}

Consider a pair of Young diagrams $\lambda$ and $\mu$ which can be joint together to form $a\times b$ rectangle (see Fig. 3):
$$\lambda_i+\mu_{b-i+1}=a,\,\, \lambda'_j+\mu'_{a-j+1}=b.$$
We will call such two diagrams {\it complementary}. 

\begin{figure}
\centerline{ \includegraphics[width=12cm]{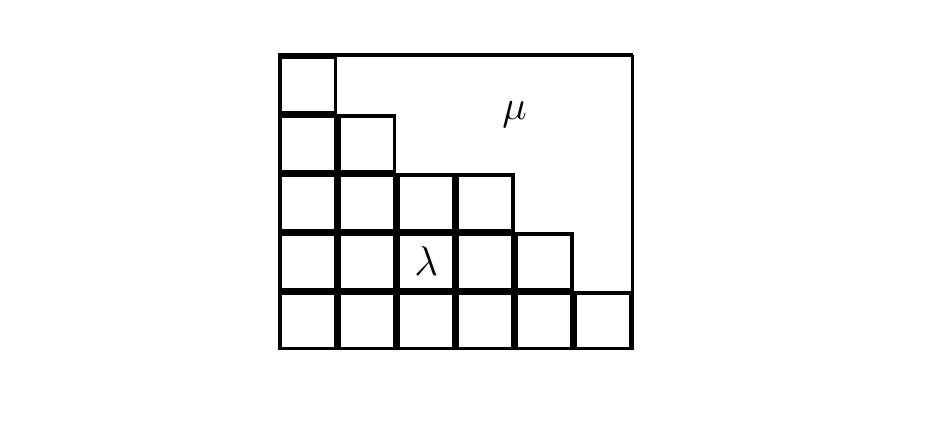}}
\caption{Complementary Young diagrams $\lambda$ and $\mu$} \label{complim}
\end{figure}

Define the following function on Young diagrams depending on two parameters $p$ and $x$:
\begin{equation}\label{var}
\varphi_p(\lambda,x)=\prod_{(i,j)\in\lambda}\frac{j-1+k(i-1-p)+x}{\lambda_i-j+k(i-1-\lambda'_j)+x}
\end{equation}
\begin{equation}\label{var2}
=\prod_{(i,j)\in\lambda}\frac{\lambda_i-j+k(i-1-p)+x}{\lambda_i-j+k(i-1-\lambda'_j)+x}
\end{equation}
with the assumption that for empty Young diagram $\varphi_p(\emptyset,x)= 1.$
Such a function was first introduced by Stanley \cite{Stanley} in the theory of Jack polynomials.

\begin{lemma} \label{lem1} For any pair of complementary Young diagrams $\lambda$ and $\mu$ forming $a\times b$ rectangle
$$\varphi_b(\lambda,x)=\varphi_b(\mu,x).$$
\end{lemma}
\begin{proof} By induction in $a+b$. If $a+b=2$ then we have $\lambda=(1),\mu=\emptyset$ or  the other way around.   Therefore 
$$
\varphi_1(\lambda,x)=\frac{-k+x}{-k+x}=1=\varphi_1(\mu,x)
$$
 Let now $a+b>2$. There are two cases: first  when $\lambda_1=a$  or $\mu_1=a$ and  the second when $\lambda_1'=b$ or $\mu'_1=b$. 
Let us consider the first case. By symmetry we can assume that   $\lambda_1=a$ and set $\nu=\lambda\setminus\lambda_1$.  Then we have 
$$
\frac{\varphi_b(\lambda)}{\varphi_{b-1}(\nu)}=\prod_{(i,j)\in\lambda}\frac{j-1+k(i-1-b)+x}{\lambda_i-j+k(i-1-\lambda'_j)+x}\prod_{(i,j)\in\nu}\frac{\nu_i-j+k(i-1-\nu'_j)+x}{j-1+k(i-b)+x}
$$
$$
=\varphi_{b}((\lambda_1),x)\varphi_{b-1}(\nu,x)\varphi_{b-1}(\nu,x)^{-1}=\varphi_{b}((\lambda_1),x)=\prod_{j=1}^a\frac{j-1-kb+x}{a-j-k\lambda'_j+x}.
$$
Now
$$
\frac{\varphi_b(\mu)}{\varphi_{b-1}(\mu)}=\prod_{(i,j)\in\mu}\frac{j-1+k(i-1-b)+x}{\mu_i-j+k(i-1-\mu'_j)+x}\prod_{(i,j)\in\mu}\frac{\mu_i-j+k(i-1-\mu'_j)+x}{j-1+k(i-b)+x}
$$
$$
=\prod_{(i,j)\in\mu}\frac{j-1+k(i-1-b)+x}{j-1+k(i-b)+x}=\prod_{j=1}^a\prod_{i=1}^{\mu'_j}\frac{j-1+k(i-1-b)+x}{j-1+k(i-b)+x}
$$
$$
=\prod_{j=1}^a\frac{j-1-kb+x}{j-1+k(\mu'_j-b)+x}=\prod_{j=1}^a\frac{j-1-kb+x}{a-j-k\lambda'_j+x}.
$$
where in the last row  we have made the change $j\rightarrow a-j+1$ and use the equality $\lambda'_j+\mu'_{a-j+1}=b$ in the denominator.
Thus we see that 
$$\frac{\varphi_b(\lambda)}{\varphi_{b-1}(\nu)}=\frac{\varphi_b(\mu)}{\varphi_{b-1}(\mu)}.$$ Since by induction $\varphi_{b-1}(\nu)={\varphi_{b-1}(\mu)}$ we have $\varphi_b(\lambda)=\varphi_b(\mu)$ in that case.

Consider now the second case. Set $\nu=\lambda\setminus\lambda'_1$.
As before we can assume that    $\lambda'_1=b$. By inductive hypothesis  $\varphi_b(\mu,x)=\varphi_b(\nu,x).$ 
The equality $\varphi_{b}(\nu,x)=\varphi_b(\lambda,x)$ is clear from the second expression (\ref{var2}) for $\varphi_b(\lambda,x)$.
\end{proof}

\begin{lemma} \label{lem2} Let $\lambda,\mu$ be two partitions, $a\ge\mu_1$ and $N\ge l(\lambda)+l(\mu)$ and 
$\nu=(\lambda_1+a, \dots, \lambda_{r}+a, a,\dots, a, a-\mu_s, \dots, a-\mu_1) \in \mathcal P$.
Then 
\begin{equation}\label{lemm2}
\varphi_{N}(\nu,x)=\varphi_{N}(\lambda,x)\varphi_{N}(\mu,x)\varphi_{N}(\lambda,\mu,x)
\end{equation}
where $\varphi_{p}(\lambda,\mu,x)$ is given by the formula 
\begin{equation}\label{phiphi}
\varphi_{p}=\prod_{i=1}^{l(\lambda)}\prod_{j=1}^{l(\mu')}\frac{\lambda_i+j-1+k(i-1-p)+x}{j-1+k(i-1-p)+x}\frac{j-1+k(i-1+\mu'_j-p)+x}{\lambda_i+j-1+k(i-1+\mu'_j-p)+x}.
\end{equation}
\end{lemma}
\begin{proof} Let $\tau \subset \lambda$ be some subset of $\lambda \in \mathcal P$. Define $\psi_{p}(\lambda,\tau,x)$ similarly to $\varphi_p(\lambda,x)$ as
\begin{equation}\label{psi}
\psi_p(\lambda,\tau,x)=\prod_{(i,j)\in\tau}\frac{\lambda_i-j+k(i-1-p)+x}{\lambda_i-j+k(i-1-\lambda'_j)+x}.
\end{equation}
Split the Young diagram $\nu$ into three parts as follows:
$$
\varphi_{N}(\nu,x)=\psi_{N}(\nu,\nu_1,x)\psi_{N}(\nu,\nu_2,x)\psi_{N}(\nu,\nu_3,x)
$$
where
$$
\nu_{1}=\{(ij)\mid1\le i\le\l(\lambda),\,a+1\le j\le a+\lambda_i\}
$$
$$
\nu_2=\{(ij)\mid1\le i\le l(\lambda),\,1\le j\le a\}
$$
$$
\nu_{3}=\{(ij)\mid 1\le j\le a,\,l(\lambda)+1\le i\le N-\mu'_{a-j+1}\}
$$
(see Fig. 4).

\begin{figure}
\centerline{ \includegraphics[width=10cm]{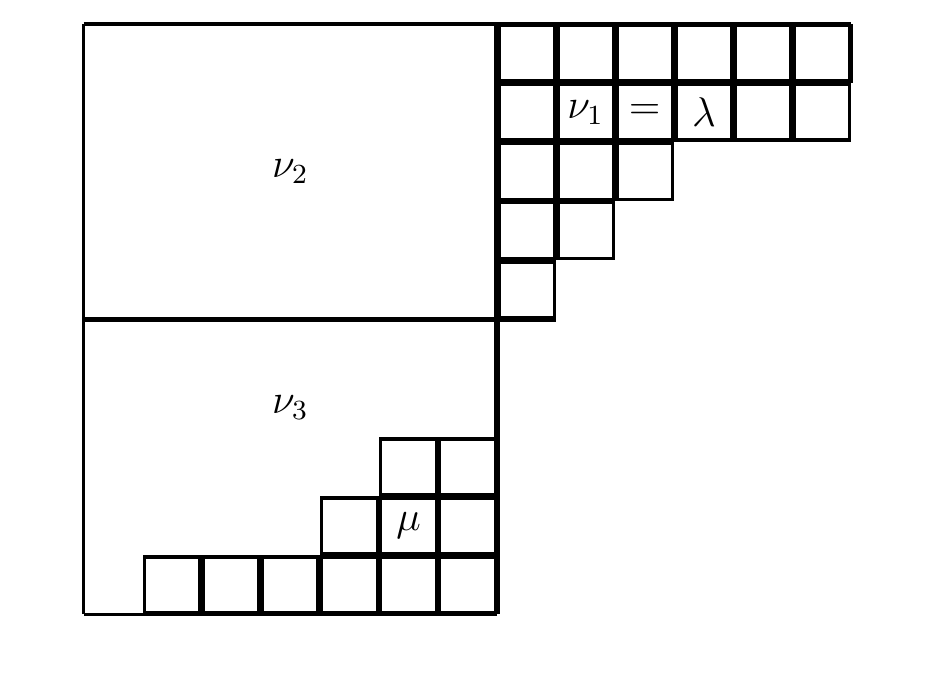}}
\caption{Three parts of the Young diagram $\nu$.} \label{3parts}
\end{figure}

We have (using the second formula (\ref{var2}) for $\varphi_{p}(\lambda,x)$) 
$$
\varphi_{N}(\nu,x)=\psi_{N}(\nu,\nu_1,x)\psi_{N}(\nu,\nu_2,x)\psi_{N}(\nu,\nu_3,x).
$$
It is easy to see that
$$
\psi_{N}(\nu,\nu_1,x)=\varphi_{N}(\lambda,x),
$$
$$
\psi_{N}(\nu,\nu_3,x)=\psi_{N}(\nu_2\cup\nu_3,\nu_3,x)=\frac{\varphi_{N}(\nu_2\cup\nu_3,x)}{\psi_{N}(\nu_2\cup\nu_3,\nu_2,x)}
$$
But according to lemma \ref{lem1} we have $\varphi_{N}(\nu_2\cup\nu_3,x)=\varphi_{N}(\mu,x)$ since $\nu_2\cup\nu_3$ is complementary to $\mu.$  
We have also
$$
\psi_N(\nu,\nu_2,x)=\prod_{(ij)\in\nu_2}\frac{\nu_i-j+k(i-1-N)+x}{\nu_i-j+k(i-1-\nu'_j)+x}=
$$
$$
\prod_{i=1}^{l(\lambda)}\prod_{j=1}^{a}\frac{\lambda_i+a-j+k(i-1-N)+x}{\lambda_i+a-j+k(i-1-N+\mu'_{a-j+1})+x}=
$$
$$
\prod_{i=1}^{l(\lambda)}\prod_{j=1}^{a}\frac{\lambda_i+j-1+k(i-1-N)+x}{\lambda_i+j-1+k(i-1-N+\mu'_j)+x}
$$
where we made the change $j\rightarrow a-j+1.$ 
Now we only need to compute  $ \psi_{N}(\nu_2\cup\nu_3,\nu_2,x)$. But this  product we can get from the previous one by setting $\lambda_i=0$ to have
$$
 \psi_{N}(\nu_2\cup\nu_3,\nu_2,x)=\prod_{i=1}^{l(\lambda)}\prod_{j=1}^{a}\frac{j-1+k(i-1-N)+x}{j-1+k(i-1-N+\mu_j)+x}.
$$
Taking $a=l(\mu')$  we have the claim.
\end{proof}

Now we are ready to prove the main result of this section. 

\begin{thm} (Evaluation Theorem) Let $\varepsilon_{p_0}$ be the homomorphism $ \Lambda^\pm \rightarrow \Bbb C$ defined by $\varepsilon(p_i)=p_0,\,i\in\Bbb Z$. Then the evaluation of the Jack--Laurent symmetric function $P^{(k,p_0)}_{\lambda,\mu}$ can be given by
\begin{equation}\label{eval}
\varepsilon_{p_0}(P^{(k,p_0)}_{\lambda,\mu})=\varphi_{p_0}(\lambda,0)\varphi_{p_0}(\mu,0)\varphi_{p_0}(\lambda,\mu,0),
\end{equation}
where the functions $\varphi_{p}(\lambda,x)$, $\varphi_{p}(\lambda,\mu,x)$ are defined by formulae (\ref{var}), (\ref{phiphi}).
\end{thm}
\begin{proof}  
Denote by $r(\lambda,\mu,k)(p_0)$ the right hand side of the formula (\ref{eval}). According to Stanley \cite{Stanley} for the usual  Jack polynomials $P^{(k)}_{\lambda}$ we have $$\varepsilon_{p_0}(P^{(k)}_{\lambda})=\varphi_{p_0}(\lambda,0).$$ For fixed  $\lambda,\mu$  and $k$ (assumed to be generic) the evaluation $\varepsilon_{p_0}(P^{(k,p_0)}_{\lambda,\mu})$ is a rational function of $p_0.$ We need to prove that $\varepsilon_{p_0}(P^{(k,p_0)}_{\lambda,\mu})=r(\lambda,\mu,k)(p_0).$ Since both sides are rational functions we only need to verify this for large enough integers $p_0=N\in\Bbb Z_{>0}.$ But in that case we have
$$
\varepsilon_{N}(P^{(k,N)}_{\lambda,\mu})=\varepsilon_{N}(P_{\nu}(k,N))=\varphi_{N}(\nu,0)=r(\lambda,\mu,k)(N)
$$
according to lemma \ref{lem2}. 
\end{proof}

\begin{corollary} Jack--Laurent symmetric functions $P^{(k,p_0)}_{\alpha}$ satisfy the $\theta$-duality
$$
\theta^{-1}(P^{(k,p_0)}_{\lambda,\mu})=d_{\lambda,\mu}(k,p_0) P^{(k^{-1},kp_0)}_{\lambda',\mu'},
$$
where  as before $\theta$ is defined by $\theta(p_a)= k  p_a$ and 
\begin{equation}\label{duad}
d_{\lambda,\mu}(k,p_0)=\frac{\varepsilon_{p_0}(P^{(k,p_0)}_{\lambda,\mu})}{\varepsilon_{kp_0}(P^{(k^{-1},kp_0)}_{\lambda',\mu'})}.
\end{equation}
\end{corollary}
\begin{proof} 
We know from Corollary \ref{cor22} that $\theta^{-1}(P^{(k,p_0)}_{\lambda,\mu})=d_{\lambda,\mu} P^{(k^{-1},kp_0)}_{\lambda',\mu'}$ for some constants 
$d_{\lambda,\mu}.$ Applying to both sides the evaluation homomorphism $\varepsilon_{kp_0}$ and using $\theta^{-1}(p_i)=k^{-1}p_i$ we have
$$
\varepsilon_{kp_0}\circ\theta^{-1}(P^{(k,p_0)}_{\lambda,\mu})=\varepsilon_{p_0}(P^{(k,p_0)}_{\lambda,\mu}),
$$
and thus
$$
\varepsilon_{p_0}(P^{(k,p_0)}_{\lambda,\mu})=d_{\lambda,\mu} \varepsilon_{kp_0}(P^{(k^{-1},kp_0)}_{\lambda',\mu'}),
$$
which implies (\ref{duad}).
\end{proof}

\section{Symmetric bilinear form}

Let us fix the parameter $k$, which we assume in this section to be negative real.

We start with the finite-dimensional case. The original CMS operator is clearly formally self-adjoint with respect to the standard scalar product
$$(\psi_1, \psi_2)=\int \psi_1(z) \bar{\psi_2}(z) dz$$
with the standard Lebesgue measure $dz=dz_1\dots dz_N$ on $\mathbb R^N.$ After the gauge $\psi=f \Psi_0$ and change $x_j=e^{2z_j}$ we naturally come to the following symmetric bilinear form for the Laurent polynomials $f,g \in \Lambda_N^{\pm}$
\begin{equation}\label{scalla}
(f,g):=c_N(k) \int_{T^N}f(x)g^*(x) \Delta_N (x, k) dx
\end{equation}
where  $T^N$ is the complex torus with $|x_j|=1, \ i=1, \dots, N$, $dx$ is the Haar measure on $T^N,$ $g^*(x_1,\dots, x_N)=g(x^{-1}_1,\dots, x^{-1}_N)$, and
\begin{equation}\label{delt}
\Delta_N(x,k)=\prod_{i,j:j\neq i}^N(1-x_i/x_j)^{-k}
\end{equation}
(cf. Macdonald \cite{Ma}, p.383, who is using parameter $\alpha=-1/k$). 
The normalisation constant $c_N(k)$ is chosen in such a way that $(1,1)_N=1$:
$$c_N^{-1}(k)=\int_{T^N}\Delta_N (x, k) dx.$$
Note that for negative real $k$ the integral (\ref{scalla}) is clearly convergent for all Laurent polynomials $f,g$
and that on the Laurent polynomials with real coefficients (in particular, for Jack--Laurent polynomials with real $k$) the product (\ref{scalla}) coincides with the Hermitian scalar product
$$(f,g):=c_N(k) \int_{T^N}f(x)\bar g(x) \Delta_N (x, k) dx.$$

Since the eigenfunctions of a self-adjoint operator are orthogonal the Jack polynomials $P^{(k)}_{\lambda}(x_1,\dots,x_N)$
 are orthogonal with respect to the product (\ref{scalla}). Using formulae (10.37), (10.22) from \cite{Ma} we have
$$
(P^{(k)}_{\lambda}, P^{(k)}_{\lambda})=\prod_{(i,j)\in \lambda}\frac{\lambda_i-j-k(\lambda_j'-i)+1}{\lambda_i-j-k(\lambda_j'-i)-k}
\frac{j-1+k(i-1)-kN}{j+ki-kN},
$$ 
 which can be rewritten in our notations as
\begin{equation}\label{nota}
(P^{(k)}_{\lambda}, P^{(k)}_{\lambda})=\frac{\varphi_{N}(\lambda,0)}{\varphi_{N}(\lambda,1+k)}.
\end{equation}

We can extend now this formula to the Jack--Laurent polynomials $P^{(k)}_{\chi}$ for any integer non-decreasing sequence $\chi=(\chi_1, \dots, \chi_N)$ by adding a large $a$ to all its parts to make them positive. Note that both $\varphi_{N}(\lambda,0)$ and $\varphi_{N}(\lambda,1+k)$ do not change under this operation and that the integral $\int_{T^N}f(x)f^*(x) \Delta (x, k) dx$ is invariant under $f(x) \rightarrow (x_1\dots x_N)^a f(x)$, so this procedure is well-defined.

Now let's look at the infinite-dimensional case.

\begin{thm} 
There exists a unique symmetric bilinear form $( \, , \,)_{p_0}$ on $\Lambda^{\pm}$ rationally depending on $p_0$ such that
Jack--Laurent symmetric functions $P^{(k,p_0)}_{\alpha}$ are orthogonal and
\begin{equation}\label{sca}
(\varphi_N(P^{(k,N)}_{\alpha}),\varphi_N(P^{(k,N)}_{\alpha}))=(P^{(k,N)}_{\alpha},P^{(k,N)}_{\alpha})_N
\end{equation}
for all sufficiently large $N,$ where the product in the left hand side is defined by (\ref{scalla}). The corresponding square norm of the Jack--Laurent symmetric function $P^{(k,p_0)}_{\alpha}$ with bipartition $\alpha=(\lambda,\mu)$ is equal to
\begin{equation}\label{sca2}
(P^{(k,p_0)}_{\alpha}, P^{(k,p_0)}_{\alpha})_{p_0}=\frac{\varphi_{p_0}(\lambda,0)\varphi_{p_0}(\mu,0)\varphi_{p_0}(\lambda,\mu,0)}{\varphi_{p_0}(\lambda,1+k)\varphi_{p_0}(\mu,1+k)\varphi_{p_0}(\lambda,\mu,1+k)}.
\end{equation}
\end{thm}

\begin{proof}  
The uniqueness is obvious since the rational function is determined by its values at sufficiently large integers.

To prove the existence we simply check that the formula (\ref{sca2}) defines the symmetric bilinear form satisfying (\ref{sca}). 
We have according to (\ref{jlsf}) that 
$$\varphi_N(P^{(k,N)}_{\alpha})=(x_1\dots x_N)^{-a} P^{(k)}_{\nu}(x_1,\dots, x_N),
$$
so that
$$
 (\varphi_N(P^{(k,N)}_{\alpha}),\varphi_N(P^{(k,p_0)}_{\alpha}))=(P^{(k)}_{\nu},P^{(k)}_{\nu}).
 $$
By (\ref{nota}) we have 
$$
(P^{(k)}_{\nu}, P^{(k)}_{\nu})=\frac{\varphi_{N}(\nu,0)}{\varphi_{N}(\nu,1+k)},
$$ 
 which by formula (\ref{lemm2}) from lemma \ref{lem2} coincides with the right hand side of (\ref{sca2}) for $p_0=N.$
 \end{proof}
 
Note that in contrast to the usual Jack case \cite{Ma} the bilinear form $( \, , \,)_{p_0}$ is {\it not positive definite} on real Laurent symmetric functions, as it follows from (\ref{sca2}). In order to have positive definite form one should send $p_0$ to infinity, see the last section.
 
\section{Special case $k=-1:$ Schur--Laurent symmetric functions}

The case  $k=-1$ is very important for representation theory of Lie superalgebra $\mathfrak{gl}(n,m)$ (see \cite{MVDJ2,CW}).
In this case the corresponding Jack--Laurent symmetric functions (whose existence is not obvious) do not depend on $p_0$, as one can see already in the simplest case
$$P^{(k,p_0)}_{1,1}= p_1 p_{-1} - \frac{p_0}{1+k-kp_0}.$$

 \begin{proposition}\label{slam}
The limit $$S_{\lambda,\mu}:=\lim_{k \to -1} P^{(k,p_0)}_{\lambda,\mu}$$
does exist for generic $p_0$ and does not depend on $p_0.$ 
\end{proposition}

We call  $S_{\lambda,\mu}$ the {\it  Schur--Laurent symmetric functions}. The image of these functions under the homomorphism $\varphi_N$ coincide with the {\it symmetric Schur polynomials $s_{\bar \mu, \lambda}(x)$ indexed by a composite partition} $\bar \mu; \lambda$ (see \cite{Moens} for a brief history of these polynomials and their role in representation theory).  Here are the two simplest examples of Schur-Laurent symmetric functions
$$S_{1,1}=p_1p_{-1}-1, \,\,\, S_{1^2,1}=\frac{1}{2}(p_1^2-p_2)p_{-1}-p_1.$$

\begin{proof}  Use the induction on $|\lambda|$. When $\lambda =\emptyset$ then $P^{(k,p_0)}_{\emptyset,\mu}=P^{(k)*}_{\mu}$, where 
$P^{(k)*}_{\mu}$ is the usual Jack symmetric function. It is well known that $P^{(-1)}_{\mu}$ is well-defined (recall that $k=-1$ corresponds to $\alpha=1$ in Jack's notations) and coincide with Schur symmetric function $S_{\mu}$ (see e.g \cite{Ma}). 

To prove the induction step one can use the Pieri formula (\ref{pieri1}). The left hand side is well defined at $k=-1$ by induction assumption. Restrict the CMS operator $\mathcal L_{k, p_0}$ onto the invariant subspace generated by the linear combinations of the Jack--Laurent symmetric functions in the right hand side of Pieri formula for generic values of the parameters. One can check analysing proof of Theorem 3.1 that for  $k=-1$ and generic $p_0$ the corresponding eigenvalues $E_1, \dots, E_k$ are distinct. This means that the component $V(x, \alpha)P_{\lambda+x,\mu}=Q(\mathcal L_{k, p_0})(p_1P_{\lambda,\mu})$ with polynomial 
$$Q(E)=C \prod_{j=2}^k (E-E_j), \,\, C^{-1}=\prod_{j=2}^k (E_1-E_j),$$ 
where $E_1$ is the eigenvalue corresponding to $P_{\lambda+x,\mu}.$ Since $\mathcal L_{k, p_0}$ is polynomial in parameters and $E_1\neq E_j$
the product $V(x, \alpha)P_{\lambda+x,\mu}$ is well defined for $k=-1$ and generic $p_0.$ Since the coefficients $V(x, \alpha)$ tend to 1 when $k \to -1$ this means that $P_{\lambda+x,\mu}$ is well-defined as well. This proves the existence of Schur--Laurent symmetric functions for generic $p_0$; their independence on $p_0$ follows from the Laurent version of Jacobi--Trudy formula below.
 \end{proof}

Let $h_i \in \Lambda \subset \Lambda^{\pm}, \,i \in \mathbb Z$ be the complete symmetric functions \cite{Ma} for $i \geq 0$ and $h_i =0$ for $i<0.$ Define 
$h_i^*$ as the image of $h_i$ under the $*$-involution in $\Lambda^{\pm}.$

 \begin{thm} \label{main123}
 The Schur-Laurent symmetric functions $S_{\alpha}, \, \alpha=(\lambda,\mu)$ can be given by the following Jacobi--Trudy formula
 as $(r+s)\times(r+s)$ determinant, where $r=l(\lambda),\, s=l(\mu)$ are the number of parts in $\lambda$ and $\mu$:
\begin{equation}
\label{JT}
S_{\alpha}=\left|\begin{array}{cccc}
   h^*_{\mu_{s}}&h^{*}_{\mu_{s}-1}& \ldots &h^{*}_{\mu_{s}-s-r+1}\\
\vdots&\vdots&\ddots&\vdots\\
  h^*_{\mu_{1}+s-1}&h^{*}_{\mu_{1}+s-2}& \ldots &h^{*}_{\mu_{1}-r}\\
  h_{\lambda_{1}-s}&h_{\lambda_{1}-s+1}& \ldots &h_{\lambda_{1}+r-1}\\
\vdots&\vdots&\ddots&\vdots\\
  h_{\lambda_{r}-s-r+1}&h_{\lambda_{r}-s-r+2}& \ldots &h_{\lambda_{r}}\\
 \end{array}\right|
 \end{equation}
\end{thm}

\begin{proof}  
Applying the homomorphisms $\varphi_N: \Lambda^{\pm} \rightarrow \Lambda_N^{\pm}$ we have in the left hand side
by definition $$\varphi_N(S_{\lambda,\mu})=(x_1\dots x_N)^a S_{\nu}(x_1, \dots, x_N),$$
where $\nu_i=\chi_N(\alpha)_i +a$ with any integer $a\ge \mu_1$ and $\chi_N(\alpha)$ defined by (\ref{chiN}).
Now the proof follows from the results of Cummins and King (see formulae (3.7), (3.8) in \cite{CK} or (1.21),(1.23) in \cite{Moens}), 
who used the language of composite Young diagrams. 
\end{proof}

\section{Some conjectures and open questions}

The usual Jack symmetric functions can be defined using the following scalar product in $\Lambda$ defined in the standard basis $p_{\lambda}=p_{\lambda_1}p_{\lambda_2}\dots$  by
$$<p_{\lambda}, p_{\mu}>=(-k)^{-l(\lambda)} \prod_{j\geq 1} j^{m_j}m_j! \delta_{\lambda,\mu},$$
where $m_j$ is the number of parts of $\lambda$ equal to $j$ (see \cite{Ma}, p. 305).
It is known (see e.g. \cite{Ma}, p. 383) that this scalar product is the limit of the scalar product (\ref{sca2}) restricted on $\Lambda$ when $p_0 \rightarrow \infty.$ An interesting question is what happens on $\Lambda^{\pm}.$ 

We believe that the limit $( \, , \,)_{\infty}$ of the indefinite bilinear form (\ref{sca2}) does exist and is positive definite on real Laurent symmetric functions for real negative $k.$ 
More precisely, we conjecture that the limits of the Jack--Laurent symmetric functions $$P^{(k,\infty)}_{\alpha}:=\lim_{p_0 \to \infty}P^{(k,p_0)}_{\alpha}$$  exist for all $k \notin \mathbb Q.$ Then by (\ref{sca2}) they would provide an orthogonal basis in $\Lambda^{\pm}$ with 
\begin{equation}\label{sca3}
(P^{(k,\infty)}_{\alpha}, P^{(k,\infty)}_{\alpha})_{\infty}=\Phi(\lambda,k)\Phi(\mu,k),
\end{equation}
where
\begin{equation}\label{Phi}
\Phi(\lambda,k)=\prod_{(i,j)\in\lambda}\frac{\lambda_i-j+1+k(i-\lambda'_j)}{\lambda_i-j+k(i-1-\lambda'_j)},
\end{equation}
which can be checked to be positive for $k<0$ and all $\lambda.$

Note that in contrast to the Jack polynomial case the Laurent polynomials  $p_{\lambda,\mu}=p_{\lambda}p_{\mu}$, as well as the products of Jack symmetric functions $P^{(k)}_{\lambda}P^{(k)*}_{\mu},$ are {\it not orthogonal} with respect to $( \, , \,)_{\infty}$. What are the transition matrices between these bases and the Jack--Laurent basis $P^{(k,\infty)}_{\lambda,\mu}$ ?

In the theory of Jack symmetric functions \cite{Ma} it is known that the product $A(\lambda)P^{(k)}_{\lambda}$ 
with $$A(\lambda)=\prod_{(ij)\in \lambda}(\lambda_i-j+k(i-1-\lambda'_j))$$ depends on $k$ polynomially.
We conjecture that a similar fact is true for Jack--Laurent symmetric functions, namely that
the product 
\begin{equation}\label{J}
J^{(k,p_0)}_{\lambda,\mu}:=A(\lambda,\mu)A(\lambda)A(\mu)P^{(k,p_0)}_{\lambda,\mu}
\end{equation}
 is polynomial in $k$ and $p_0$, where
$$A(\lambda,\mu)=\prod_{i=1}^{l(\lambda)}\prod_{j=1}^{l(\mu')}(j-1+k(i-1-p_0))(\lambda_i+j-1+k(i-1+\mu'_j-p_0)).$$
A weaker version of this conjecture is that the product $A(\lambda,\mu)P^{(k,p_0)}_{\lambda,\mu}$ is polynomial in $p_0.$

The case of special parameters $k$ and $p_0$ with $p_0=n+k^{-1}m$ is very important for the representation theory 
of Lie superalgebras and is discussed in our paper \cite{SV8}.

\section{Acknowledgements}

This work was partially supported by the EPSRC (grant EP/J00488X/1). ANS is grateful to Loughborough University for the hospitality during the autumn semesters 2012-14. He also acknowledges the financial support from the Government of the Russian Federation within
the framework of the implementation of the 5-100 Programme Roadmap of the National Research University Higher School of Economics.

\end{document}